\RequirePackage{fix-cm}
\documentclass[twocolumn]{svjour3}          
\smartqed  
\usepackage{graphicx}
\usepackage{bbm}
\usepackage{amsmath, hyperref}
\usepackage{tikz}
\usetikzlibrary{math,calc}
\usepackage{pgfplots}
\usepackage{breakcites}
\usepackage{graphics}
\usepackage{epstopdf}
\usepackage{hhline}
\usepackage{tikz}
\usetikzlibrary{arrows,automata}
\usepackage{pgfplots}

\usepackage[absolute,overlay]{textpos}

\usepackage{hyperref}
\hypersetup{
  colorlinks=true,
  linkcolor=black,
  filecolor=red,      
  urlcolor=black,
  citecolor=blue
}

\usepackage{graphicx}
\usepackage{ragged2e}
\usepackage{tabto}
\usepackage{float}
\usepackage{multirow}
\usepackage{amsmath}
\usepackage{amssymb}
\usepackage{mathtools}

\usepackage{xcolor}
\definecolor{darkred}{RGB}{173,0,43}

\usepackage[nottoc]{tocbibind}

\usepackage[font=footnotesize]{subcaption}
\usepackage[font=footnotesize]{caption}

\usepackage{multirow}
\usepackage{multicol}

\usepackage{bm}
\usepackage{dsfont}

\newcommand\vct[1]{\bm{\mathsf{#1}}}
\newcommand\mtx[1]{\bm{\mathsf{#1}}}

\newcommand{\tran}{^{\mathsf{T}}}

\newcommand{\lp}{\left(}
\newcommand{\rp}{\right)}

\usepackage[sort,compress]{cite}

\newcommand{\cK}{\mathcal{K}}
\newcommand{\cN}{\mathcal{N}}
\newcommand{\R}{\mathbb{R}}
\newcommand{\bbE}{\mathbb{E}}

\DeclareMathOperator*{\argmin}{argmin}

\newtheorem{algorithm0}{Algorithm}

\numberwithin{equation}{section}

\journalname{Statistics and Computing}
\begin{document}

\title{Efficient reduced-rank methods for Gaussian processes with eigenfunction expansions}

\author{Philip Greengard         \and
        Michael O'Neil
}



\institute{Philip Greengard \at
              Columbia University \\
              \email{pg2118@columbia.edu}      
           \and
           Michael O'Neil \at
           Courant Institute, NYU \\
           \email{oneil@cims.nyu.edu}
}

\date{August 24, 2022}

\maketitle

\begin{abstract}
  In this work we introduce a reduced-rank algorithm for Gaussian
  process regression. Our numerical scheme converts a Gaussian process
  on a user-specified interval to its Karhunen-Lo\`eve expansion, the
  \mbox{$L^2$-optimal} reduced-rank representation. Numerical 
  evaluation of the Karhunen-Lo\`eve expansion is performed once
  during precomputation and involves computing a numerical
  eigendecomposition of an integral operator whose kernel is the
  covariance function of the Gaussian process. 
  The Karhunen-Lo\`eve expansion is independent of observed data and
  depends only on the covariance kernel and the size of the interval on which
  the Gaussian process is defined. 
  The scheme of this paper does not
  require translation invariance of the covariance kernel.  We also
  introduce a class of fast algorithms for Bayesian fitting of
  hyperparameters, and demonstrate the performance of our algorithms
  with numerical experiments in one and two dimensions. Extensions to
  higher dimensions are mathematically straightforward but suffer from
  the standard curses of high dimensions.
\end{abstract}


\newpage

\section{Introduction}
Over the past two decades there has been vast interest in modeling
with Gaussian processess (GPs) \cite{rasmus} across a range of
applications including astrophysics, epidemiology, ecology, climate
science, financial mathematics, and political science~\cite{dfm1, bda,
  stein1, gonzalvez1}. In many cases, the main limitation of Gaussian
process regression as a practical statistical tool is its prohibitive
computational cost (when the calculations are done directly).
Classical direct algorithms for Gaussian process regression with~$N$
data points incur an~$O(N^3)$ computational cost which, for many
modern problems, is unfeasible. As a result, there has been much
effort directed toward asymptotically efficient or approximate
computational methods for modeling with
Gaussian processes. 

Most of these methods~\cite{candela1, oneil1, datta1, minden1, dfm1,
  solin1, aki1} involve some sort of (fast) approximate inversion of the
covariance matrix~$\mtx{C}$ that appears in the likelihood
function~$p$ of a Gaussian process
\begin{equation}
  p(\vct{y}) \propto  \frac{1}{|\mtx{C}|^{1/2}} \,
  \exp \lp -\frac{1}{2} \vct{y}\tran \mtx{C}^{-1} \vct{y} \rp.
\end{equation}
In particular, reduced-rank algorithms approximate the~$N \times N$
covariance matrix~$\mtx{C}$ with a \emph{global} rank-$m$
factorization~$\mtx{X}\mtx{X}\tran$ such that
\begin{equation}
\Vert \mtx{C} - \mtx{X X}\tran \Vert <\epsilon,
\end{equation}
where~$\mtx{X}$ is an $N \times m$ matrix and~$\epsilon$ is some
tolerance chosen based on the application at hand. The quadratic
form~$\vct{y}\tran \mtx{C}^{-1} \vct{y}$ can then be computed in the
least-squares-sense.
Usually, reduced-rank algorithms rely on rough
approximations of the covariance matrix (such as standard Nystr\"om
methods whereby rows or columns are randomly sub-sampled) or they require certain
assumptions about the covariance kernel and distribution of
data points.
Furthermore, these low-rank approximation ideas can be used in a
\emph{locally} recursive fashion, as in the algorithms of~\cite{oneil1, minden1}, to
construct a hierarchical factorization of the covariance matrix that
allows for direct inversion in~$\mathcal O(N)$ time
for a 
reasonably general choice of covariance function. However, these
algorithms rely on the covariance matrix having a particular low-rank
structure away from the diagonal,
and can be prohibitively slow when, for certain kernels or
data, this condition is not met (or in the case when the ambient
dimension of the observations is high).

In the numerical methods of this paper, we decompose a Gaussian
process defined on an interval (or a rectangular region of $\R^d$)
into a global expansion of fixed basis functions with random coefficients
that is optimally accurate in the $L^2$-sense. Specifically, we
introduce a numerical method for approximating a continuous Gaussian
process using its Karhunen-Lo\`eve (KL) expansion~\cite{loeve1,xiu1}.
This approach is of course a global one, and does not apply any 
hierarchical compression strategy directly to the induced covariance
matrix itself. We merely provide the numerical tools to optimally
compress, to any desired precision, a Gaussian process onto a
lower-rank subspace using its associated eigenfunction expansion,
independently of where the process was sampled.
For a Gaussian process with covariance kernel~$k: \R^d \times \R^d \to
\R$, using such a KL expansion allows a Gaussian process
\begin{equation}
  \label{eq:gp}
  f \sim \mathcal{GP}(0, k)
\end{equation}
defined on a region~$D \subset \R^d$ to be reformulated using an
expansion of the form
\begin{equation}\label{765}
f(x) = \sum_{i = 1}^\infty   \alpha_i \,  \phi_i(x),
\end{equation}
where the~$\phi_i$'s are (properly scaled) eigenfunctions of the integral
operator~$\cK$ defined by
\begin{equation}\label{484}
\cK\mu(x) = \int_D k(x, y) \, \mu(y) \, dy, \qquad \text{for } x
\in D,
\end{equation}
and where the~$\alpha_i$'s are IID normal random variables. Each
of the eigenfunctions therefore satisfies the relationship
\begin{equation}
  \lambda_i \phi_i(x) = \int_D k(x,y) \, \phi_i(y) \, dy.
\end{equation}
For the sake of convenience, in what follows we will always assume
that the eigenfunctions~$\phi_i$ have been ordered according to the
magnitude of the corresponding eigenvalue:~$\lambda_1 \geq \lambda_2
\geq \lambda_3 \geq \ldots$.
Furthermore, we will assume that the covariance kernel is
square-integrable on~$D \times D$, and therefore that the integral
operator~$\cK$ is bounded and compact when acting on square integrable
functions. This situation covers most widely used covariance
kernels~\cite{riesz1955}.
The above reformulation is valid for all $x \in D$, and the infinite
expansion in terms of the~$\phi_i$'s can be truncated depending on the
desired accuracy in approximating the covariance function (in the
least-squares sense).
We refer to KL-expansion~\eqref{765} truncated at $m$ terms to be the 
order-$m$ KL-expansion
\begin{equation}\label{765b}
  \sum_{i = 1}^m   \alpha_i \,  \phi_i(x).
\end{equation}
The KL-expansion
has several advantages over other low-rank compression techniques, a
primary advantage being that it provides the  optimal compression of a Gaussian
process in the~$L^2$ sense and can be performed independent of the
distribution of sample points of the process. In particular, if
$f$ is a mean-zero Gaussian process
on an interval $[a, b] \subset \R$ with covariance function $k$,
then \emph{any} order-$m$ reduced-rank approximation of the form
\begin{equation}\label{530}
  f(x) \approx 
\alpha_1 g_1(x) + ... + \alpha_m g_m(x),
\end{equation}
where the $g_i$'s are fixed basis functions and $\alpha_i$'s are IID
$\cN(0, 1)$ random variables, 
has an effective covariance kernel~$k_n$ defined by
\begin{equation}
  \begin{aligned}
    k_n(x, y) &= \mathbb{E} \left [ f(x) \, f(y) \right] \\
    &= \bbE \left[ \lp \sum_{i=1}^m \alpha_i g_i(x) \rp \lp
  \sum_{i=1}^m \alpha_i g_i(y) \rp 
\right] \\
&=  \sum_{i,j = 1}^m g_i(x) g_j(y) \, \bbE \left[ \alpha_i \alpha_j
\right]  \\
&= \sum_{i = 1}^m g_i(x) g_i(y) .
\end{aligned}
\end{equation}
Among all order-$m$
reduced-rank Gaussian process approximations~\eqref{530}, the
effective covariance kernel of the order-$m$ KL expansion in~\eqref{765b}, 
denoted by $k_{\text{KL}}$, satisfies
\begin{equation}\label{3101}
k_{\text{KL}} = \argmin_{k'} \int_a^b \int_a^b \lp k(x, y) - k'(x, y)
\rp^2 \,  dx \, dy,
\end{equation}
where~$k$ is the exact covariance function in~\eqref{eq:gp},
see~\cite{trefethen}.
A similar approach to parameterizing
random functions is discussed in~\cite{filip2019random}, whereby the
default expansion is taken to be in terms of Chebyshev polynomials
(i.e. trigonometric polynomials) instead of the true Karhunen-Lo\`eve
expansion.

After converting a Gaussian process to its KL-expansion, performing
statistical inference is drastically simplified. For example,
 in the canonical Gaussian process
regression task with~$N$ data points $\{(x_i, y_i)\}$, the regression model
\begin{equation}
  y \sim f(x) + \epsilon,
\end{equation}
where 
\begin{equation}
  \begin{aligned}
   \epsilon \, | \, x &\sim \cN (0,\sigma^2), \\
   f  & \sim \mathcal{GP}(0, k(x, x')),
\end{aligned}
\end{equation}
has a closed-form solution that requires~$O(Nm^2)$ operations
where~$m$ is the length of the KL-expansion. 
We also introduce an algorithm
for computing posterior moments of fully Bayesian Gaussian process
regression in which we fit two hyperparameters of the covariance
function, namely the timescale and the magnitude.

The theoretical properties of KL-expansions have been well-understood
for many years. However, in applied statistics communities, the use of 
\emph{high-order approximations} of KL-expansions has been virtually nonexistent.
Presumably one reason for this is the lack of standard tools (available in
statistics-focused software packages) for numerically computing
eigendecompositions of continuous operators.  On the other hand, in
the applied mathematics and computational physics communities, there
is a large body of analysis of integral operators and numerical tools
for their discretization (see, for example, \cite{kress1}) including
finite element algorithms for computing KL-expansions~\cite{schwab1}. 
In this paper, the primary numerical tools
we exploit for computing KL-expansions belong to a well-known class
of Nystr\"om methods~\cite{yarvin} for computing eigendecompositions
of integral operators. The dominant cost of the numerical scheme is
the diagonalization of a symmetric, positive semi-definite matrix
whose dimension scales as the number of quadrature nodes needed to
accurately discretize it. This diagonalization is performed at most once
during precomputation.


Several basis function approaches have achieved popularity in the Gaussian
process community, such as the Fourier-based
\cite{lazaro2010, rahimi2007}. Like in our approach, in \cite{lazaro2010}, a basis function 
expansion is constructed such that its effective covariance kernel 
approximates some desired kernel. Their method benefits from
the fact that Fourier basis function expansions are essentially free to compute 
and have analytical properties that are well-known. The primary advantage of 
the KL-expansion over other basis function approaches, including Fourier
methods, is that the KL-expansion is an optimal compression in $L^2$,
see~\eqref{3101}. As a general matter, the cost of performing 
Gaussian process regression with basis function approaches is $O(Nm^2)$, where $N$ is
the number of data points and $m$ is the number of basis functions.
As a result, reducing the number of basis functions in a Gaussian process
representation can be crucial for practical use and can result in 
substantial computational savings. 
 
The scheme of this paper is similar in spirit to that
of~\cite{solin1}.  In~\cite{solin1}, the authors introduce a method
that approximates the KL-expansion by first representing the integral
operator~$\cK$ in~\eqref{484} as a finite linear combination of powers 
of the Laplace operator, and then subsequently approximating the 
eigenfunctions of that operator.  In this paper, we directly compute high-order
approximations to eigenfunctions of~$\cK$ with quadrature-based methods.

While the mathematical properties of KL-expansions generalize naturally
to arbitrary dimensions, their use as a practical statistical tool is limited 
to around $3$ dimensions, or $4$ for very smooth kernels. 
This is due to the standard curse of dimensionality -- 
for a given level of accuracy, the number
of basis functions needed in $d$ dimensions scales as $m^d$ where
$m$ is the number of functions needed in $1$ dimension. This exponential
scaling makes the algorithms of this paper impractical for high-dimensional environments. In particular, there are two computations that become computationally intractable.  
First, construction of the KL-expansion in $d$ dimensions would involve an eigendecomposition of 
a $m^d \times m^d$ matrix, a procedure that requires $O(m^{3d})$ operations. 
Similarly, the linear system of the Gaussian process regression 
task would require solving an $m^{d} \times m^{d}$ linear system which also
requires $O(Nm^{2d})$ operations where $N$ is the number of data 
points. It is, however, likely that the methods of this paper could be
used in conjunction with a spatially adaptive low-dimensional
approximation of the data in order to (locally) reduce the ambient
dimension of the problem.

The remainder of this paper is structured as follows. In the following
section we provide background on mathematical concepts that will be
used in subsequent sections. In Section~\ref{s531} we describe the
primary numerical scheme of this paper -- an algorithm for computing
KL-expansions. Section~\ref{sec:regression} contains a description of
how the algorithms of this paper can be used in Gaussian process
regression; we then describe an efficient algorithm for Bayesian Gaussian
process regression in Section~\ref{s550}. Sections~\ref{s560}
and~\ref{s570} contain numerical methods for computing KL-expansions for
non-smooth kernels and high-dimension Gaussian process problems,
respectively.
We
provide the results of numerical implementations of the algorithms of
this paper in Section~\ref{s580}, and lastly, in
Section~\ref{sec:conclusions} we offer some ideas regarding future
directions of these techniques as well as a discussion of the main
failure-mode of the algorithm.

\section{Mathematical apparatus}
\label{sec:math}

We start by introducing some background on Gaussian processes and
approximation theory that will be used throughout the paper.

\subsection{Gaussian Processes}
Given a mean function~$m: \R^d \to \R$ and a covariance function~$k:
\R^d \times \R^d \to \R$, a
Gaussian process is a random function $f: \R^d \to \R$, denoted by
\begin{equation}
f \sim \mathcal{GP}(m(x), k(x,x')),
\end{equation}
such that for any collection of points $x_1, \ldots, x_N \in \R^d$,
we have that
\begin{equation}
\lp f(x_1) \cdots f(x_N) \rp\tran \sim \mathcal{N}(\vct{m}, \mtx{C})
\end{equation}
where the mean $\vct{m} \in \R^N$ satisfies
\begin{equation}
\vct{m} = \lp m(x_1) \cdots m(x_N) \rp\tran,
\end{equation}
and where~$\mtx{C}$ is an~$N \times N$ covariance matrix with entries
\begin{equation}
  \label{eq:cij}
  \mtx{C}_{ij} = k(x_i, x_j).
\end{equation}
For the remainder of this paper we assume $m(x) = 0$ for
convenience. As in much of the computational Gaussian process
literature, this assumption has no impact on the methods
of this paper.
The function~$k$ must satisfy particular properties to
ensure the positivity of the underlying probability measure. Namely,
for any choice of the $x_i$'s above, the matrix~$\mtx{C}$ defined
by~\eqref{eq:cij} must be symmetric positive
semi-definite~\cite{cressie1}.  In the case where the covariance
function~$k$ is translation invariant (i.e.~$f$ is a stationary
process), $k$ is a function of  $|x-y|$, and Bochner's 
Theorem~\cite{rasmus} shows that $k$
is admissable if and only if its Fourier transform is real and
non-negatively valued. We merely point out this as a fact, but
will not make use of it explicitly in this work as our methods also apply
to kernels that are not translation invariant.

Furthermore, the following is a well-known theorem that we will in
fact exploit when discretizing the integral operator associated with
the covariance kernel of Gaussian processes~\cite{stoer}. We state the
theorem in the one-dimensional case, but it of course can be extended
analogously to arbitrary dimensions.
\begin{theorem}[Mercer's Theorem]\label{210b}
  Let $k$ be a continuous, symmetric, positive semi-definite kernel
  defined on~$[a,b] \times [a,b]$. Then the integral operator
\begin{equation}
\cK f(x) = \int_a^b k(x, x') \, f(x') \, dx'
\end{equation}
has real, non-negative eigenvalues~$\lambda_i$ with corresponding
eigenfunctions~$u_i$. We assume that the eigenfunctions 
$u_i$ have $L^2$ norm of $1$. The kernel~$k$ then can be written as
\begin{equation}
k(x, x') = \sum_{i=1}^\infty \lambda_i \, u_i(x) \, u_i(x'),
\end{equation}
where convergence is absolute and uniform.
\end{theorem}
The above theorem is merely a continuous version of the standard
finite-dimensional result for symmetric positive semi-definite
matrices.

\subsection{Karhunen-Lo\`eve expansions}
\label{sec:kl}

In this section we describe the theoretical basis of the algorithms we
use for low-rank compression of Gaussian processes. The central
analytical tool is a special case of the well-known Karhunen-Lo\`eve
theorem~\cite{xiu1}.
\begin{theorem}\label{410}(Karhunen-Lo\`eve)
  Let $f$ be a Gaussian process on $D \subset \R$ with covariance
  kernel~$k$. Then, for all $x\in D$ we have that~$f$ can be written as
\begin{equation}\label{380}
f(x) = \sum_{i=1}^\infty \alpha_i \, u_i(x)
\end{equation}
where for $i=1,2,...$, 
\begin{equation}
\alpha_i \sim \mathcal{N}\lp 0, \lambda_i\rp
\end{equation}
and the~$\lambda_i$'s and~$u_i$'s
are eigenvalues and eigenfunctions of the integral operator $\cK$ defined by
\begin{equation}\label{780}
\cK\mu(x) = \int_D k(x, x') \, \mu(x') \, dx'.
\end{equation}
\end{theorem}

We refer to expansion (\ref{380}) as a Karhunen-Lo\`eve (KL)
expansion. As before, we will assume that the eigenfunctions are
ordered in terms of non-decreasing values of the associated
eigenvalues.

The eigenfunctions $u_i$ of (\ref{380}) are assumed to have unit
$L^2$ norm. That is
\begin{align}
\int_D |u_i(x)|^2 \, dx = 1
\end{align}
for all $i$. We will be denoting by $\phi_i$ a scaling of eigenfunction $u_i$
by the square root of its eigenvalue. Specifically,
\begin{align}\label{891}
\phi_i(x) = \sqrt{\lambda_i} u_i(x). 
\end{align}

The following theorem illustrates that a truncated KL-expansion with 
IID Gaussian coefficients can be used
as a practical tool to represent a Gaussian process distribution. 
In particular, the effective 
covariance function of a finite KL-expansion is the outer product
of the eigenfunctions of $\cK$ of (\ref{780}). Furthermore convergence of 
the outer product is sufficiently fast for practical purposes
 for a large class of covariance kernels.
\begin{theorem}\label{360}
Let $\hat{f}$ be defined by the order-$m$ \\
KL-approximation
\begin{equation}\label{745a}
\hat{f}(x) = \sum_{i=1}^m \alpha_i \phi_i(x).
\end{equation}
for all $x \in [-1, 1]$ where
\begin{equation}
\alpha \sim \cN(0, \mtx{I})
\end{equation}
and $\phi_i(x) = \sqrt{\lambda_i} u_i(x)$ where $\lambda_i$ and 
$u_i$ are the eigenvalues and eigenfunctions of integral operator
(\ref{780}).  Then $\hat{f}$ is a Gaussian process with covariance kernel 
\begin{equation}\label{330}
\hat{k}_m(x, x') = \sum_{i=1}^m \lambda_iu_i(x)u_i(x').
\end{equation}
Additionally, for smooth $k$
\begin{equation}
\| k  - \hat{k}_m \|_2^2 
\end{equation} 
decays exponentially in $m$. For kernels $k$ with continuous derivatives up to order $j$, the decay of $\| k  - \hat{k}_m \|_2^2$ is no slower than $O(1/m^{j+1})$. 
\end{theorem}

\begin{proof}
For all $x$, clearly $\hat{f}(x)$ is Gaussian with mean given by,
\begin{equation}
\bbE[\hat{f}(x)] =  \sum_{i=1}^m \bbE[\alpha_i] \phi_i(x) = 0.
\end{equation}
Additionally, for all $x, x' \in [-1, 1]$, by independence of $\alpha_i$ and $\alpha_j$ for $i\neq j$, we have
\begin{equation}
  \begin{aligned}  
\bbE[\hat{f}(x)\hat{f}(x')] & =  \bbE \left[ \sum_{i=1}^m \alpha_i
  \phi_i(x) \sum_{i=1}^m \alpha_i \phi_i(x') \right] \\
& = \sum_{i=1}^m \phi_i(x)\phi_i(x')
  \end{aligned}
\end{equation}
Using (\ref{891}) we obtain
\begin{equation}\label{720}
\bbE[\hat{f}(x)\hat{f}(x')] = \sum_{i=1}^m \lambda_u u_i(x) u_j(x').
\end{equation}
Proofs of convergence rates of (\ref{720}) to $k$ can be found in, 
for example, \cite{trefethen}. 
\end{proof}

The existence of KL-expansions has been well-known since at
least the 1970s~\cite{loeve1}, however their use as a numerical
tool for Gaussian process regression has been virtually
nonexistent. This is mainly due to a lack of computing power and
numerical algorithms for computing the eigenfunctions and eigenvalues
used in~\eqref{380}. Recent advances in numerical computation,
primarily coming from the field of computational physics, has turned
the evaluation of eigendecompositions of
integral operators~(\ref{780}) into a well-understood
and computationally tractable
exercise.

We lastly note that the choice of region of integration is somewhat
arbitrary in the above theorem -- as long as the interval~$D$ contains
all observation points and points at which predictions wish to be
made, it is a suitable interval.  In the next section, we will
truncate the expansion in Theorem~\ref{410} to obtain an approximation
to the Gaussian process.  For a fixed kernel and fixed level of
accuracy increasing the size of the region~$D$ on which the Gaussian
process is defined does result in the need for a marginally larger KL
expansion. It is therefore advantageous from a computational
standpoint to choose a region that narrowly includes all points of
interest.

\section{Numerical computation of KL-expansions}
\label{s531}
In this section we describe a numerical scheme for computing the 
KL-expansion of a Gaussian process with a fixed covariance
function to any desired precision. We describe the algorithm in the
context of a Gaussian process defined on a region of $\R$, though
generalizations to higher dimensions are straightforward and in Section
\ref{s570} we provide the analogous algorithm for two-dimensional Gaussian
processes. For now, the
interval is chosen to be~$[-1, 1]$ out of convenience -- any interval
$[a, b]$ can be exchanged with $[-1, 1]$ along with the corresponding
transformation of Gaussian nodes and weights.

The algorithm consists mainly of computing
eigenfunctions and eigenvalues of the integral operator
$\cK: L^2[-1,1] \rightarrow L^2[-1, 1]$ defined by
\begin{equation}\label{210}
\cK f(x) = \int_{-1}^1 k(x, x') \, f(x') \, dx',
\end{equation}
where $k: [-1, 1]^2 \to \R$ is a covariance kernel. 
The numerical scheme discretizes the integral operator $\cK$ and represents 
the action of the integral operator on a function as a
matrix-vector multiplication. The eigenfunctions and eigenvalues
of~$\cK$ are then approximated with the eigenvectors and eigenvalues of
the matrix approximation to~$\cK$. The algorithm is well-known, and is a
slight variant of the algorithm contained in Section 4.3
of~\cite{yarvin}.
\begin{algorithm0}[Evaluation of KL-expansion] \, 
  \label{a1}
  \begin{enumerate}
  \item 
We start by constructing the $n \times n$ matrix $A$ defined by 
\begin{equation}\label{790}
\mtx{A}_{i, j} = \sqrt{w_iw_j}\, k(x_i, x_j)
\end{equation}
where 
\begin{equation}\label{735}
x_1,...,x_n
\end{equation}
denote the order-$n$ Gaussian nodes
\begin{equation}
w_1,...,w_n
\end{equation}
the order-$n$ Gaussian weights.
\item Compute the diagonal form of the symmetric matrix $\mtx{A}$. That is, find 
the orthogonal matrix $\mtx{U}$ and the diagonal matrix $\mtx{D}$ such that
\begin{equation}\label{370}
\mtx{A} = \mtx{UDU\tran}.
\end{equation}
We denote the $i$-th entry of the diagonal of $\mtx{D}$ by $\lambda_i$. 
\item Construct the $n \times n$ matrix $\hat{\mtx{U}} = [u_i] $ defined by 
\begin{equation}
\hat{\mtx{U}}_{i, j} = \mtx{U}_{i, j}/\sqrt{w_i}.
\end{equation}

\item Convert the eigenfunction approximations in $\hat{\mtx{U}}$ to a matrix $\mtx{A}$
of Legendre expansions. Do this by applying to $\hat{\mtx{U}}$ the 
matrix $\mtx{M}$ (see Theorem \ref{230}) 
that converts tabulations at Gaussian nodes to Legendre coefficients:
\begin{equation}
\mtx{A} = \mtx{M}\hat {\mtx{U}}.
\end{equation}

\item Evaluate the eigenfunction approximations $u_i : [-1, 1] \to \R$ by the formula
\begin{equation}\label{340}
u_i(x) = \sum_{j=1}^n \mtx{A}_{j, i} P_{j-1}(x)
\end{equation}
for all $x \in [-1, 1]$ and $i=1,2,...,k$ where $P_j$ denotes the order-$j$ Legendre polynomial. 

\item \label{671} Scale the eigenfunctions $u_i$ by the square root of the eigenvalues. That is, we define $\phi_i$ by
\begin{equation}\label{350}
\phi_i(x) = \sqrt{\lambda_i} u_i(x).
\end{equation}

\item The KL-expansion of length $m \leq n$ is given by
\begin{equation}\label{622}
\hat{f}(x) = \alpha_1 \phi_1(x) + \alpha_2\phi_2(x) + ... + \alpha_m \phi_m(x)
\end{equation}
for all $x \in [-1, 1]$ where $\alpha_i \sim \cN(0, 1)$ are IID Gaussian random variables. 
\end{enumerate}
\end{algorithm0}

We note that the scaling of the eigenfunctions in step \ref{671} of Algorithm
\ref{a1} is not strictly necessary, but enforces that the coefficients of the 
KL-expansion are all $\cN(0, 1)$ and consequently that Gaussian process
regression is the standard ridge regression. Figure \ref{2200} includes
plots of eigenfunctions $\phi_i$ of (\ref{350}) for a squared exponential
kernel. 

The computational cost of Algorithm \ref{a1} is $O(n^3)$ where $n$ 
is the number of discretization nodes. In the following section we describe
theoretical and numerical considerations for choosing $n$.

\begin{figure}[t!]
\centering
\begin{tikzpicture}[scale=1.0]

\begin{axis}[
    xmin=-1, xmax=1,
    ymin=-3, ymax=3,
    xtick={-1,-0.5,0.0,0.5,1},
    xlabel=$x$,
    ytick={-3,-2,-1,0,1,2,3},
    ylabel=$\phi_i(x)$,
    legend pos= south west
]
 
\addplot[dashed, color=red, line width=0.3mm]
    coordinates 
    {
(   -1.0000000000000000,   -0.1845141940154184)
(   -0.9797979797979798,   -0.2039445398692058)
(   -0.9595959595959596,   -0.2241885826268916)
(   -0.9393939393939394,   -0.2451494832175868)
(   -0.9191919191919192,   -0.2667245259759273)
(   -0.8989898989898990,   -0.2888076802120987)
(   -0.8787878787878788,   -0.3112920962513386)
(   -0.8585858585858586,   -0.3340724259428587)
(   -0.8383838383838383,   -0.3570468751196029)
(   -0.8181818181818181,   -0.3801189179800603)
(   -0.7979797979797980,   -0.4031986287745423)
(   -0.7777777777777778,   -0.4262036123205416)
(   -0.7575757575757576,   -0.4490595396699985)
(   -0.7373737373737373,   -0.4717003169360915)
(   -0.7171717171717171,   -0.4940679325366104)
(   -0.6969696969696970,   -0.5161120401307993)
(   -0.6767676767676767,   -0.5377893410671376)
(   -0.6565656565656566,   -0.5590628314736479)
(   -0.6363636363636364,   -0.5799009758757726)
(   -0.6161616161616161,   -0.6002768623773481)
(   -0.5959595959595959,   -0.6201673851037757)
(   -0.5757575757575757,   -0.6395524889271946)
(   -0.5555555555555556,   -0.6584145005308507)
(   -0.5353535353535352,   -0.6767375595142682)
(   -0.5151515151515151,   -0.6945071541605194)
(   -0.4949494949494949,   -0.7117097591082131)
(   -0.4747474747474747,   -0.7283325666866952)
(   -0.4545454545454545,   -0.7443633000733063)
(   -0.4343434343434343,   -0.7597900945484226)
(   -0.4141414141414141,   -0.7746014326826424)
(   -0.3939393939393939,   -0.7887861199597983)
(   -0.3737373737373737,   -0.8023332887765309)
(   -0.3535353535353535,   -0.8152324206439606)
(   -0.3333333333333333,   -0.8274733784764912)
(   -0.3131313131313130,   -0.8390464428738903)
(   -0.2929292929292928,   -0.8499423481358548)
(   -0.2727272727272727,   -0.8601523153034335)
(   -0.2525252525252525,   -0.8696680807599612)
(   -0.2323232323232323,   -0.8784819198461661)
(   -0.2121212121212120,   -0.8865866655781175)
(   -0.1919191919191918,   -0.8939757229472410)
(   -0.1717171717171716,   -0.9006430794804748)
(   -0.1515151515151515,   -0.9065833127972767)
(   -0.1313131313131313,   -0.9117915958651914)
(   -0.1111111111111110,   -0.9162637005659814)
(   -0.0909090909090908,   -0.9199960000699089)
(   -0.0707070707070706,   -0.9229854703981971)
(   -0.0505050505050504,   -0.9252296914467103)
(   -0.0303030303030303,   -0.9267268476549584)
(   -0.0101010101010101,   -0.9274757284362564)
(    0.0101010101010102,   -0.9274757284362581)
(    0.0303030303030305,   -0.9267268476549633)
(    0.0505050505050506,   -0.9252296914467142)
(    0.0707070707070707,   -0.9229854703981987)
(    0.0909090909090911,   -0.9199960000699060)
(    0.1111111111111112,   -0.9162637005659747)
(    0.1313131313131315,   -0.9117915958651852)
(    0.1515151515151516,   -0.9065833127972751)
(    0.1717171717171717,   -0.9006430794804774)
(    0.1919191919191920,   -0.8939757229472449)
(    0.2121212121212122,   -0.8865866655781188)
(    0.2323232323232325,   -0.8784819198461631)
(    0.2525252525252526,   -0.8696680807599549)
(    0.2727272727272729,   -0.8601523153034286)
(    0.2929292929292930,   -0.8499423481358518)
(    0.3131313131313131,   -0.8390464428738892)
(    0.3333333333333335,   -0.8274733784764919)
(    0.3535353535353536,   -0.8152324206439630)
(    0.3737373737373739,   -0.8023332887765335)
(    0.3939393939393940,   -0.7887861199597966)
(    0.4141414141414144,   -0.7746014326826356)
(    0.4343434343434345,   -0.7597900945484124)
(    0.4545454545454546,   -0.7443633000732989)
(    0.4747474747474749,   -0.7283325666866945)
(    0.4949494949494950,   -0.7117097591082177)
(    0.5151515151515154,   -0.6945071541605218)
(    0.5353535353535355,   -0.6767375595142635)
(    0.5555555555555556,   -0.6584145005308425)
(    0.5757575757575759,   -0.6395524889271883)
(    0.5959595959595960,   -0.6201673851037742)
(    0.6161616161616164,   -0.6002768623773497)
(    0.6363636363636365,   -0.5799009758757749)
(    0.6565656565656568,   -0.5590628314736470)
(    0.6767676767676769,   -0.5377893410671331)
(    0.6969696969696970,   -0.5161120401307944)
(    0.7171717171717173,   -0.4940679325366089)
(    0.7373737373737375,   -0.4717003169360949)
(    0.7575757575757578,   -0.4490595396700006)
(    0.7777777777777779,   -0.4262036123205389)
(    0.7979797979797982,   -0.4031986287745388)
(    0.8181818181818183,   -0.3801189179800595)
(    0.8383838383838385,   -0.3570468751196022)
(    0.8585858585858588,   -0.3340724259428575)
(    0.8787878787878789,   -0.3112920962513353)
(    0.8989898989898992,   -0.2888076802120961)
(    0.9191919191919193,   -0.2667245259759286)
(    0.9393939393939394,   -0.2451494832175813)
(    0.9595959595959598,   -0.2241885826268937)
(    0.9797979797979799,   -0.2039445398692022)
(    1.0000000000000000,   -0.1845141940154173)
};

 \addplot[
       dash dot, color=black,  line width=0.3mm
    ]
    coordinates 
    {
(   -1.0000000000000000,   -0.3472800669374392)
(   -0.9797979797979798,   -0.3810462645202798)
(   -0.9595959595959596,   -0.4155848959752597)
(   -0.9393939393939394,   -0.4506167779764693)
(   -0.9191919191919192,   -0.4858506098302390)
(   -0.8989898989898990,   -0.5209890629388604)
(   -0.8787878787878788,   -0.5557347289691020)
(   -0.8585858585858586,   -0.5897956844172507)
(   -0.8383838383838383,   -0.6228904675838989)
(   -0.8181818181818181,   -0.6547523130131219)
(   -0.7979797979797980,   -0.6851325436912361)
(   -0.7777777777777778,   -0.7138030780139774)
(   -0.7575757575757576,   -0.7405580622560871)
(   -0.7373737373737373,   -0.7652146862221857)
(   -0.7171717171717171,   -0.7876132771013034)
(   -0.6969696969696970,   -0.8076167926077082)
(   -0.6767676767676767,   -0.8251098487559518)
(   -0.6565656565656566,   -0.8399974206535750)
(   -0.6363636363636364,   -0.8522033479471554)
(   -0.6161616161616161,   -0.8616687620933446)
(   -0.5959595959595959,   -0.8683505328458718)
(   -0.5757575757575757,   -0.8722198087056797)
(   -0.5555555555555556,   -0.8732607028416657)
(   -0.5353535353535352,   -0.8714691540542131)
(   -0.5151515151515151,   -0.8668519731455955)
(   -0.4949494949494949,   -0.8594260694841275)
(   -0.4747474747474747,   -0.8492178410070302)
(   -0.4545454545454545,   -0.8362627033722762)
(   -0.4343434343434343,   -0.8206047300781308)
(   -0.4141414141414141,   -0.8022963745309253)
(   -0.3939393939393939,   -0.7813982465489189)
(   -0.3737373737373737,   -0.7579789189119602)
(   -0.3535353535353535,   -0.7321147436227593)
(   -0.3333333333333333,   -0.7038896619571786)
(   -0.3131313131313130,   -0.6733949966978271)
(   -0.2929292929292928,   -0.6407292188529062)
(   -0.2727272727272727,   -0.6059976844732705)
(   -0.2525252525252525,   -0.5693123398156978)
(   -0.2323232323232323,   -0.5307913950635144)
(   -0.2121212121212120,   -0.4905589681691362)
(   -0.1919191919191918,   -0.4487447012226726)
(   -0.1717171717171716,   -0.4054833521866723)
(   -0.1515151515151515,   -0.3609143649782585)
(   -0.1313131313131313,   -0.3151814208241669)
(   -0.1111111111111110,   -0.2684319736424342)
(   -0.0909090909090908,   -0.2208167719787248)
(   -0.0707070707070706,   -0.1724893697892710)
(   -0.0505050505050504,   -0.1236056281438611)
(   -0.0303030303030303,   -0.0743232097363712)
(   -0.0101010101010101,   -0.0248010679427551)
(    0.0101010101010102,    0.0248010679427553)
(    0.0303030303030305,    0.0743232097363732)
(    0.0505050505050506,    0.1236056281438644)
(    0.0707070707070707,    0.1724893697892738)
(    0.0909090909090911,    0.2208167719787265)
(    0.1111111111111112,    0.2684319736424347)
(    0.1313131313131315,    0.3151814208241688)
(    0.1515151515151516,    0.3609143649782635)
(    0.1717171717171717,    0.4054833521866800)
(    0.1919191919191920,    0.4487447012226811)
(    0.2121212121212122,    0.4905589681691423)
(    0.2323232323232325,    0.5307913950635184)
(    0.2525252525252526,    0.5693123398157007)
(    0.2727272727272729,    0.6059976844732746)
(    0.2929292929292930,    0.6407292188529111)
(    0.3131313131313131,    0.6733949966978344)
(    0.3333333333333335,    0.7038896619571879)
(    0.3535353535353536,    0.7321147436227706)
(    0.3737373737373739,    0.7579789189119714)
(    0.3939393939393940,    0.7813982465489255)
(    0.4141414141414144,    0.8022963745309258)
(    0.4343434343434345,    0.8206047300781290)
(    0.4545454545454546,    0.8362627033722780)
(    0.4747474747474749,    0.8492178410070397)
(    0.4949494949494950,    0.8594260694841427)
(    0.5151515151515154,    0.8668519731456089)
(    0.5353535353535355,    0.8714691540542200)
(    0.5555555555555556,    0.8732607028416675)
(    0.5757575757575759,    0.8722198087056821)
(    0.5959595959595960,    0.8683505328458804)
(    0.6161616161616164,    0.8616687620933579)
(    0.6363636363636365,    0.8522033479471683)
(    0.6565656565656568,    0.8399974206535830)
(    0.6767676767676769,    0.8251098487559545)
(    0.6969696969696970,    0.8076167926077096)
(    0.7171717171717173,    0.7876132771013110)
(    0.7373737373737375,    0.7652146862221997)
(    0.7575757575757578,    0.7405580622560977)
(    0.7777777777777779,    0.7138030780139799)
(    0.7979797979797982,    0.6851325436912372)
(    0.8181818181818183,    0.6547523130131270)
(    0.8383838383838385,    0.6228904675839043)
(    0.8585858585858588,    0.5897956844172535)
(    0.8787878787878789,    0.5557347289691024)
(    0.8989898989898992,    0.5209890629388619)
(    0.9191919191919193,    0.4858506098302476)
(    0.9393939393939394,    0.4506167779764658)
(    0.9595959595959598,    0.4155848959752693)
(    0.9797979797979799,    0.3810462645202786)
(    1.0000000000000000,    0.3472800669374407)
    };

\addplot[color=blue, line width=0.3mm
    ]
    coordinates 
    {
(   -1.0000000000000000,   -0.4709012587129028)
(   -0.9797979797979798,   -0.5102710734924099)
(   -0.9595959595959596,   -0.5490572350485798)
(   -0.9393939393939394,   -0.5867097383332728)
(   -0.9191919191919192,   -0.6226712890924553)
(   -0.8989898989898990,   -0.6563895549525427)
(   -0.8787878787878788,   -0.6873290720812459)
(   -0.8585858585858586,   -0.7149823686999383)
(   -0.8383838383838383,   -0.7388799344740417)
(   -0.8181818181818181,   -0.7585987507685062)
(   -0.7979797979797980,   -0.7737691931669074)
(   -0.7777777777777778,   -0.7840802164536486)
(   -0.7575757575757576,   -0.7892828258775186)
(   -0.7373737373737373,   -0.7891919205486052)
(   -0.7171717171717171,   -0.7836866605312657)
(   -0.6969696969696970,   -0.7727095557840531)
(   -0.6767676767676767,   -0.7562645017429134)
(   -0.6565656565656566,   -0.7344139940373353)
(   -0.6363636363636364,   -0.7072757460349705)
(   -0.6161616161616161,   -0.6750189111229368)
(   -0.5959595959595959,   -0.6378600808908594)
(   -0.5757575757575757,   -0.5960591947932485)
(   -0.5555555555555556,   -0.5499154602218642)
(   -0.5353535353535352,   -0.4997633473708719)
(   -0.5151515151515151,   -0.4459686931798987)
(   -0.4949494949494949,   -0.3889249244710917)
(   -0.4747474747474747,   -0.3290493927996394)
(   -0.4545454545454545,   -0.2667798024365973)
(   -0.4343434343434343,   -0.2025707076667388)
(   -0.4141414141414141,   -0.1368900552107943)
(   -0.3939393939393939,   -0.0702157508800313)
(   -0.3737373737373737,   -0.0030322353195226)
(   -0.3535353535353535,    0.0641729392362595)
(   -0.3333333333333333,    0.1309125318421854)
(   -0.3131313131313130,    0.1967030289486894)
(   -0.2929292929292928,    0.2610680411614459)
(   -0.2727272727272727,    0.3235416630973647)
(   -0.2525252525252525,    0.3836717750984908)
(   -0.2323232323232323,    0.4410232631655335)
(   -0.2121212121212120,    0.4951811322387158)
(   -0.1919191919191918,    0.5457534877145747)
(   -0.1717171717171716,    0.5923743606501030)
(   -0.1515151515151515,    0.6347063532795852)
(   -0.1313131313131313,    0.6724430830830722)
(   -0.1111111111111110,    0.7053114055551049)
(   -0.0909090909090908,    0.7330733979158092)
(   -0.0707070707070706,    0.7555280882015136)
(   -0.0505050505050504,    0.7725129164126480)
(   -0.0303030303030303,    0.7839049166481181)
(   -0.0101010101010101,    0.7896216113985942)
(    0.0101010101010102,    0.7896216113985971)
(    0.0303030303030305,    0.7839049166481236)
(    0.0505050505050506,    0.7725129164126522)
(    0.0707070707070707,    0.7555280882015131)
(    0.0909090909090911,    0.7330733979158042)
(    0.1111111111111112,    0.7053114055550987)
(    0.1313131313131315,    0.6724430830830670)
(    0.1515151515151516,    0.6347063532795815)
(    0.1717171717171717,    0.5923743606500995)
(    0.1919191919191920,    0.5457534877145704)
(    0.2121212121212122,    0.4951811322387106)
(    0.2323232323232325,    0.4410232631655260)
(    0.2525252525252526,    0.3836717750984806)
(    0.2727272727272729,    0.3235416630973512)
(    0.2929292929292930,    0.2610680411614321)
(    0.3131313131313131,    0.1967030289486769)
(    0.3333333333333335,    0.1309125318421734)
(    0.3535353535353536,    0.0641729392362457)
(    0.3737373737373739,   -0.0030322353195400)
(    0.3939393939393940,   -0.0702157508800495)
(    0.4141414141414144,   -0.1368900552108114)
(    0.4343434343434345,   -0.2025707076667549)
(    0.4545454545454546,   -0.2667798024366154)
(    0.4747474747474749,   -0.3290493927996635)
(    0.4949494949494950,   -0.3889249244711190)
(    0.5151515151515154,   -0.4459686931799249)
(    0.5353535353535355,   -0.4997633473708947)
(    0.5555555555555556,   -0.5499154602218868)
(    0.5757575757575759,   -0.5960591947932745)
(    0.5959595959595960,   -0.6378600808908897)
(    0.6161616161616164,   -0.6750189111229724)
(    0.6363636363636365,   -0.7072757460350096)
(    0.6565656565656568,   -0.7344139940373737)
(    0.6767676767676769,   -0.7562645017429475)
(    0.6969696969696970,   -0.7727095557840864)
(    0.7171717171717173,   -0.7836866605313080)
(    0.7373737373737375,   -0.7891919205486568)
(    0.7575757575757578,   -0.7892828258775670)
(    0.7777777777777779,   -0.7840802164536885)
(    0.7979797979797982,   -0.7737691931669481)
(    0.8181818181818183,   -0.7585987507685498)
(    0.8383838383838385,   -0.7388799344740847)
(    0.8585858585858588,   -0.7149823686999771)
(    0.8787878787878789,   -0.6873290720812790)
(    0.8989898989898992,   -0.6563895549525729)
(    0.9191919191919193,   -0.6226712890924918)
(    0.9393939393939394,   -0.5867097383332890)
(    0.9595959595959598,   -0.5490572350486086)
(    0.9797979797979799,   -0.5102710734924194)
(    1.0000000000000000,   -0.4709012587129276)
    };
    
\addplot[dotted, color=purple,  line width=0.3mm]
    coordinates 
    {
(   -1.0000000000000000,    0.5454555417691086)
(   -0.9797979797979798,    0.5804649918472718)
(   -0.9595959595959596,    0.6123676951414146)
(   -0.9393939393939394,    0.6403614505999484)
(   -0.9191919191919192,    0.6636744578108956)
(   -0.8989898989898990,    0.6815871120914829)
(   -0.8787878787878788,    0.6934525731661582)
(   -0.8585858585858586,    0.6987153408032883)
(   -0.8383838383838383,    0.6969271924540439)
(   -0.8181818181818181,    0.6877599892019967)
(   -0.7979797979797980,    0.6710150245347933)
(   -0.7777777777777778,    0.6466287625439191)
(   -0.7575757575757576,    0.6146749759560027)
(   -0.7373737373737373,    0.5753634396206300)
(   -0.7171717171717171,    0.5290354541218363)
(   -0.6969696969696970,    0.4761565625218733)
(   -0.6767676767676767,    0.4173068794906521)
(   -0.6565656565656566,    0.3531694776683353)
(   -0.6363636363636364,    0.2845172747981369)
(   -0.6161616161616161,    0.2121988422941106)
(   -0.5959595959595959,    0.1371235176496106)
(   -0.5757575757575757,    0.0602461557241420)
(   -0.5555555555555556,   -0.0174481967689255)
(   -0.5353535353535352,   -0.0949594685256197)
(   -0.5151515151515151,   -0.1712873819079395)
(   -0.4949494949494949,   -0.2454457685351137)
(   -0.4747474747474747,   -0.3164762575183508)
(   -0.4545454545454545,   -0.3834612326250570)
(   -0.4343434343434343,   -0.4455359686075933)
(   -0.4141414141414141,   -0.5018998653467434)
(   -0.3939393939393939,   -0.5518267030870199)
(   -0.3737373737373737,   -0.5946738446358937)
(   -0.3535353535353535,   -0.6298903125429818)
(   -0.3333333333333333,   -0.6570236722514027)
(   -0.3131313131313130,   -0.6757256569457970)
(   -0.2929292929292928,   -0.6857564768648641)
(   -0.2727272727272727,   -0.6869877654167277)
(   -0.2525252525252525,   -0.6794041264661596)
(   -0.2323232323232323,   -0.6631032613721359)
(   -0.2121212121212120,   -0.6382946703169768)
(   -0.1919191919191918,   -0.6052969396793390)
(   -0.1717171717171716,   -0.5645336451312880)
(   -0.1515151515151515,   -0.5165279182675505)
(   -0.1313131313131313,   -0.4618957424269342)
(   -0.1111111111111110,   -0.4013380605229607)
(   -0.0909090909090908,   -0.3356317938068455)
(   -0.0707070707070706,   -0.2656198852477528)
(   -0.0505050505050504,   -0.1922004943966614)
(   -0.0303030303030303,   -0.1163154820161474)
(   -0.0101010101010101,   -0.0389383322665745)
(    0.0101010101010102,    0.0389383322666201)
(    0.0303030303030305,    0.1163154820161938)
(    0.0505050505050506,    0.1922004943967077)
(    0.0707070707070707,    0.2656198852477973)
(    0.0909090909090911,    0.3356317938068877)
(    0.1111111111111112,    0.4013380605229998)
(    0.1313131313131315,    0.4618957424269723)
(    0.1515151515151516,    0.5165279182675877)
(    0.1717171717171717,    0.5645336451313250)
(    0.1919191919191920,    0.6052969396793734)
(    0.2121212121212122,    0.6382946703170063)
(    0.2323232323232325,    0.6631032613721579)
(    0.2525252525252526,    0.6794041264661750)
(    0.2727272727272729,    0.6869877654167380)
(    0.2929292929292930,    0.6857564768648716)
(    0.3131313131313131,    0.6757256569458030)
(    0.3333333333333335,    0.6570236722514057)
(    0.3535353535353536,    0.6298903125429823)
(    0.3737373737373739,    0.5946738446358882)
(    0.3939393939393940,    0.5518267030870080)
(    0.4141414141414144,    0.5018998653467254)
(    0.4343434343434345,    0.4455359686075720)
(    0.4545454545454546,    0.3834612326250350)
(    0.4747474747474749,    0.3164762575183296)
(    0.4949494949494950,    0.2454457685350939)
(    0.5151515151515154,    0.1712873819079185)
(    0.5353535353535355,    0.0949594685255967)
(    0.5555555555555556,    0.0174481967689017)
(    0.5757575757575759,   -0.0602461557241660)
(    0.5959595959595960,   -0.1371235176496331)
(    0.6161616161616164,   -0.2121988422941342)
(    0.6363636363636365,   -0.2845172747981612)
(    0.6565656565656568,   -0.3531694776683578)
(    0.6767676767676769,   -0.4173068794906687)
(    0.6969696969696970,   -0.4761565625218883)
(    0.7171717171717173,   -0.5290354541218568)
(    0.7373737373737375,   -0.5753634396206562)
(    0.7575757575757578,   -0.6146749759560256)
(    0.7777777777777779,   -0.6466287625439351)
(    0.7979797979797982,   -0.6710150245348098)
(    0.8181818181818183,   -0.6877599892020170)
(    0.8383838383838385,   -0.6969271924540648)
(    0.8585858585858588,   -0.6987153408033098)
(    0.8787878787878789,   -0.6934525731661735)
(    0.8989898989898992,   -0.6815871120915010)
(    0.9191919191919193,   -0.6636744578109228)
(    0.9393939393939394,   -0.6403614505999573)
(    0.9595959595959598,   -0.6123676951414411)
(    0.9797979797979799,   -0.5804649918472814)
(    1.0000000000000000,   -0.5454555417691372)
    };

    \legend{$i=1$, $i=2$, $i=3$, $i=4$}

\end{axis}

\end{tikzpicture}

\caption{The basis functions $\phi_i$ (see (\ref{350})) of integral operator $\cK: L^2[-1, 1] \to L^2[-1, 1]$ for $i=1,2,3,4$ where $k$ is a squared-exponential kernel with $\ell =0.2$.}
\label{2200}
\end{figure}
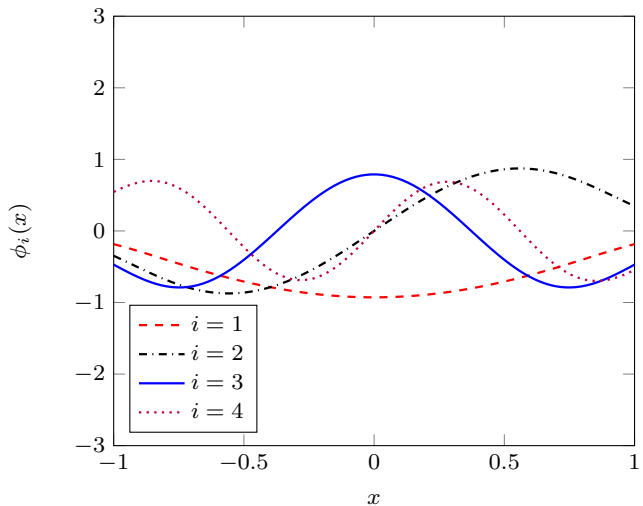

\subsection{Error control}
\label{s120}
Suppose that using Algorithm \ref{a1} with $n$ nodes we construct the 
approximate order-$m$ KL-expansion 
\begin{align}\label{812}
\alpha_1\phi_1(x) + ... + \alpha_m \phi_m(x)
\end{align}
for some $m \leq n$. A natural metric for measuring the error of expansion (\ref{812})
is the $L^2$ difference between the true covariance kernel~$k$ and 
the effective covariance kernel of (\ref{812}). We define this error
to be  $\epsilon_n$. That is, 
\begin{equation}\label{612}
\epsilon_n = \left\Vert k(x, x') - \sum_{i=1}^m \lambda_i u_i(x)u_i(x') \right\Vert_2
\end{equation}
(see Theorem \ref{360}) where $\lambda_i$ and $u_i$
are the eigenvalue and eigenfunction approximations of~(\ref{350}). There are 
two sources of error that contribute to $\epsilon_n$:
\begin{enumerate}
\item 
\textbf{Discretization error}: The eigenvalues and eigenvectors
used in (\ref{812}) are approximated numerically with Algorithm \ref{a1}. 
For kernels that have continuous derivatives of order $j$, 
the convergence of those approximations in $n$, the number of nodes, is 
approximately $O(n^{-(j+1)})$~\cite{yarvin}. 
More precisely, for all fixed $m$, we define $\alpha_n$ by 
\begin{align}\label{832}
\alpha_n = \bigg\| \sum_{i=1}^m \lambda_i u_i(x)u_i(x') - 
\sum_{i=1}^m \lambda_i^n u_i^n(x) u_i^n(x') \bigg\|_2
\end{align}
where $u_i$ and $\lambda_i$ are the exact eigenfunctions and eigenvalues
and $\lambda_i^n$ and $u_i^n$ are the approximations obtained
via Algorithm \ref{a1} with $n$ nodes. Then $\alpha_n = O(n^{-(j+1)})$
independent of $m$.

\item
\textbf{Truncation error}: Suppose that for all $i \leq m$, the eigenvalues and 
eigenfunctions of (\ref{812}) are obtained to infinite precision. Then error 
$\epsilon_n$ of (\ref{612}) becomes
\begin{align}\label{814}
\bigg\| k(x, x') - \sum_{i=1}^m \lambda_i u_i(x)u_i(x') \bigg\|_2 
= \bigg( \sum_{i=m+1}^\infty \lambda_i^2 \bigg)^{1/2}.
\end{align}
Equation (\ref{814}), combined with the $L^2$ optimality of the 
eigenfunction expansion (see (\ref{3101})), shows that for any basis function 
Gaussian process 
regression algorithm, an expansion of length $m$ will have an $L^2$ error of 
at least 
\begin{equation}
\bigg( \sum_{i=m+1}^\infty \lambda_{i+1}^2 \bigg)^{1/2}.
\end{equation}
If the kernel has $j$ times continuous derivatives, then the magnitude 
of $\ell$-th eigenvalue will be approximately 
$O(\ell^{-(j+1)})$. For those kernels, 
\begin{equation}\label{814b}
\begin{split}
\bigg\| k(x, x') - \sum_{i=1}^m \lambda_i u_i(x)u_i(x') \bigg\|_2 
& = \bigg( \sum_{i=m+1}^\infty \lambda_i^2 \bigg)^{1/2} \\
& = O\big(m^{(-4j+1)/2}\big).
\end{split}
\end{equation}

\end{enumerate}
A further discussion of the accuracy of Algorithm \ref{a1} can be found 
in~\cite{yarvin}. In Section \ref{s580} we provide numerical evaluations of 
$\epsilon_n$ of (\ref{612}) for Mat\'ern and squared-exponential kernels. 

For Gaussian processes over $\R^d$, $\epsilon_n$ is an 
integral over a region of $\R^{2d}$ and can be computed with 
adaptive Gaussian quadrature. For $d > 1$, evaluation of these
integrals can be computationally costly. 
For a more tractable alternative to computing 
(\ref{612}) directly, we use
the following measurement of error. 
We first approximate the discretization error by running 
Algorithm \ref{a1} with $n$ nodes. We denote the eigenvalue approximations
\begin{equation}
\lambda_1^n,...,\lambda_m^n.
\end{equation}
We then repeat the same procedure with 
$2n$ nodes and obtain eigenvalue approximations
\begin{equation}
\lambda_1^{2n},...,\lambda_m^{2n}.
\end{equation}
We then check the maximum difference between the $\lambda_{i}^n$ and 
$\lambda_{i}^{2n}$. That is, we evaluate $\delta_{max}$
where 
\begin{equation}
\delta_{max} = \max_{i \leq m} {|\lambda_i^n - \lambda_i^{2n}|}
\end{equation}
The maximum of $\delta_{max}$ and $\lambda_{m+1}$ can be 
used as a proxy for (\ref{832}). The order of magnitude of the 
$L^2$ error $\epsilon_n$ of the approximate KL-expansion can 
therefore be approximated by $\delta_{max} + m^{(-4j+1)/2}$.

We note that for a given level of accuracy, the number of terms needed
to achieve that accuracy depends on the ratio of the size of the 
region where the Gaussian process is defined 
($[a, b]$ in (\ref{780})) and the timescale of 
the kernel. In Figure \ref{2003all} we provide plots of the eigenvalues 
for the squared exponential 
and Mat\'ern kernels in one and two dimensions. 

In Section \ref{s580}, we demonstrate the performance of Algorithm 
\ref{a1} in Gaussian process regression problems in $\R$ and $\R^2$. 
Notably, for commonly-used kernels, the costs of computing KL-expansions 
are negligible compared to the costs of 
performing statistical inference in problems with even moderate 
amounts of data. In Tables \ref{3021} and \ref{3031} we provide the 
accuracy of KL-expansions computed using Algorithm \ref{a1} as a 
function of the number of nodes $n$ (see (\ref{735})) for two commonly 
used covariance kernels -- squared exponential and Mat\'ern. We 
measure accuracy of the KL-expansion when using $n$ nodes as 
the $L^2$ difference between the true kernel and the effective kernel 
of the order-$n$ KL-expansion.

\begin{figure*}
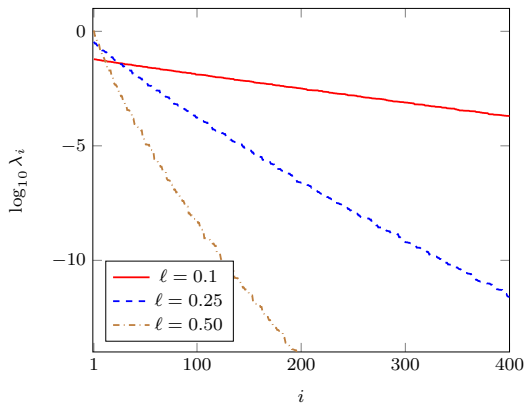
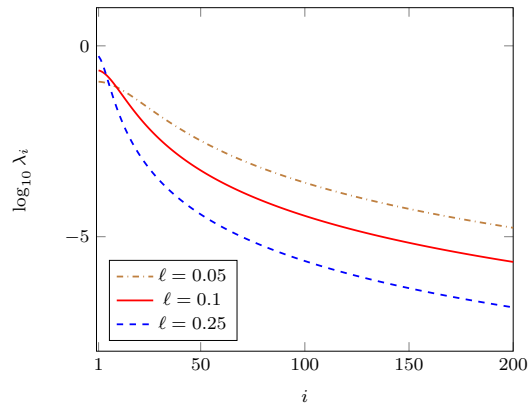
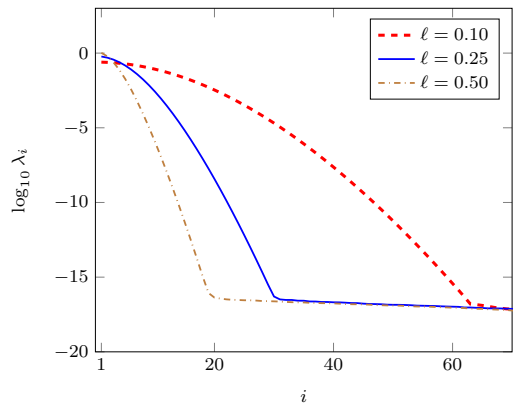

\begin{subfigure}{0.49\linewidth}


\caption{$\log_{10}(\lambda_i)$ where $\lambda_i$ (see (\ref{370})) are 
eigenvalues of the 
integral operator $\cK: L^2([-1, 1] \times [-1, 1]) \to L^2 ([-1, 1] \times [-1, 1])$ 
with squared-exponential kernel and various $\ell$ for Gaussian processes
on $\R^2$.}
\label{2001}
\end{subfigure}
\qquad
\begin{subfigure}{0.49\linewidth}
\centering
%

\caption{$\log_{10}(\lambda_i)$ where $\lambda_i$ (see (\ref{370})) are eigenvalues of integral operator $\cK: L^2[-1, 1] \to L^2[-1, 1]$ with Mat\'ern  kernel with $\nu=3/2$ and various $\ell$
for Gaussian processes on $\R$.}
\label{2002}
\end{subfigure}

\vspace{1cm}

\begin{subfigure}{\linewidth}
\centering
%

\caption{$\log_{10}(\lambda_i)$ where $\lambda_i$ (see (\ref{370})) are 
eigenvalues of the integral operator $\cK: L^2[-1, 1] \to 
L^2[-1, 1]$ where $k$ is a squared-exponential kernel with various $\ell$
for Gaussian processes on $\R$.}
\label{2003}
\end{subfigure}
\caption{Decay of eigenvalues for integral operators with various
  covariance kernels.}
  \label{2003all}
\end{figure*}

\section{Reduced-rank regression}
\label{sec:regression}
Representing a Gaussian process as its KL-expansion has a number of
computational and statistical advantages in problems with large amounts of data. In the canonical Gaussian process regression a user is given
data and noisy observations $\{(x_i,y_i)\}$ and seeks an unknown function $f$
under the model
\begin{equation}
  \begin{aligned}
y\, |\, x & \sim \mathcal{N}(f(x), \sigma^2) \\
f & \sim \mathcal{GP}(0, k(x, x')).
  \end{aligned}
\end{equation}
Using KL-expansions computed via Algorithm \ref{a1} we numerically convert the Gaussian process 
\begin{equation}
f \sim \mathcal{GP}(0, k(x, x'))
\end{equation}
to the KL-expansion 
\begin{equation}\label{730}
f \approx \alpha_1 \phi_1(x) + ... + \alpha_m \phi_m(x)
\end{equation}
where $f$ is defined on some user-specified region, $\alpha_i$ are IID
$\mathcal{N}(0, 1)$ random variables and $\phi_i$ are the scaled
eigenfunctions (\ref{350}). In Figure \ref{2200} we include plots of
the eigenfunctions $\phi_i$ for a squared exponential
kernel~\cite{rasmus} in one dimension.

After converting a Gaussian process to a KL-expansion, regression
tasks involve an additional $O(Nm^2)$ operations where $N$ is the
number of data points and $m$ is the size of the KL-expansion. We now 
describe two methods for statistical inference using KL-expansions 
-- inference at a set of points and a basis function 
approach~\cite{rasmus}.

\subsection{Prediction}

In many applied Gaussian process settings a user is given a set of
noisy measurements $y_i = f(x_i) + \epsilon_i$ for $i=1,...,N$, where
$x_i \in [a, b]$ are independent variables and $\epsilon_i$ is IID
Gaussian noise. After specifying a covariance function, $k$, the goal
 is usually to determine, given $\{(x_i,y_i)\}$, the
conditional (or posterior) distribution at some point 
$\tilde{x}$ or set of points in the region $[a, b]$. 
The conditional distribution of $f(\tilde{x})$
is the Gaussian
\begin{equation}
f(\tilde{x}) \, | \, \vct{x},\vct{y} \sim \mathcal{N}(\tilde{\mu}, \tilde{\sigma}^2)
\end{equation}
where $\vct{x} = (x_1,...,x_N)$, $\vct{y} = (y_1,...,y_N)$, and 
\begin{equation}
  \begin{aligned}
& \tilde{\mu} = k(\tilde{x}, \vct{x}) (\mtx{K} + \sigma^2\mtx{I})^{-1} \vct{y}  \\
& \tilde{\sigma}^2 = k(\tilde{x}, \tilde{x}) - k(\tilde{x}, \vct{x}) (\mtx{K} + \sigma^2\mtx{I})^{-1} k(\vct{x}, \tilde{x}),
  \end{aligned}
\end{equation}
where $\mtx{K}_{i, j} = k(x_i, x_j)$, $k(\tilde{x}, \vct{x})$ is the
row vector 
\begin{align}
[k(\tilde{x}, x_1), ..., k(\tilde{x}, x_N)]
\end{align}
and
$k(\tilde{x}, \vct{x}) = k(\vct{x}, \tilde{x})\tran$. After computing
the KL-expansion of the Gaussian process on $[a, b]$ with covariance
kernel $k$, we $\mtx{K}$ can be approximated via
\begin{equation}
\mtx{K} \approx \mtx{XX\tran} 
\end{equation}
where $\mtx{X}$ is the $N \times m$ matrix with entries
\begin{equation}\label{745}
\mtx{X}_{ij} = \phi_{j}(x_i).
\end{equation}
That is
\[
\mtx{X} = 
\begin{bmatrix}
    \phi_{1}(x_1) & \phi_2(x_1) & \dots  & \phi_{m}(x_1) \\
    \phi_{1}(x_2) & \phi_2(x_2) & \dots & \phi_{m}(x_2) \\
    \vdots & \vdots & \dots & \vdots \\
    \phi_{1}(x_N) & \phi_2(x_N) & \dots & \phi_{m}(x_N)
\end{bmatrix}.
\]
We note that this global low-rank approximation is equivalent to 
approximating each element $\mtx{K}_{i j} = k(x_i, x_j)$ with 
the outerproduct of eigenfunctions
\begin{equation}
\sum_{k=1}^m \phi_k(x_i)\phi_k(x_j) = \sum_{k=1}^m \lambda_k u_k(x_i)u_k(x_j),
\end{equation}
where $\lambda_i$ and $u_i$
are the eigenvalues and eigenfunctions of (\ref{350}).

We can then construct an approximation to 
$(\mtx{K} + \sigma^2 \mtx{I})^{-1}$ in $O(Nm^2)$ operations. 
This can be done by, for example, computing the SVD of $\mtx{X}$:
\begin{equation}\label{437}
\mtx{X} = \mtx{U D V}\tran,
\end{equation}
where $\mtx{U}$ is a $N \times m$ matrix with orthonormal columns, $\mtx{D}$ is a $m \times m$ diagonal matrix, and $\mtx{V}$ is a $m \times m$ orthogonal matrix. We can then use $\mtx{XX\tran}$ as the rank-$m$ approximation to $\mtx{K}$ in order to approximate 
$(\mtx{K} + \sigma^2\mtx{I})^{-1}$ via the following formula
\begin{equation}
\begin{split}
(\mtx{K} + \sigma^2\mtx{I})^{-1} 
 \approx (\mtx{XX\tran} + \sigma^2 \mtx{I})^{-1}
& = (\mtx{UD^2U\tran} + \sigma^2 \mtx{I})^{-1} \\
&  = \mtx{U}(\mtx{D^2} + \sigma^2 \mtx{I})^{-1} \mtx{U\tran}.
\end{split}
\end{equation} 
This method of constructing a global low-rank approximation to the
covariance matrix is also discussed in \cite{solin1}.

\subsection{Weight-space inference}
In addition to facilitating global low rank approximations, 
KL-expansions have the advantage that they allow for statistical 
inference in the coefficients of a basis function expansion
(the weight-space view of \cite{rasmus}).
In fact, from a computational 
standpoint, inference over coefficients is of negligible cost 
once the SVD of $\mtx{X}$ is obtained. 

From a basis function perspective, the standard Gaussian process 
regression model is the canonical $\ell^2$-regularized (ridge) 
linear regression  
\begin{equation}\label{190}
\begin{split}
\vct{y} & \sim \cN(\mtx{X}\vct{\beta}, \sigma^2) \\
\vct {\beta} & \sim \cN(0, \mtx{I})
\end{split}
\end{equation}
where $\mtx{X}$ is the $N \times m$ matrix defined in (\ref{745}). In
this model, we perform inference on $\vct{\beta}$, the coefficients in
the expansion of basis functions $\phi_i(x)$, see (\ref{350}). The
corresponding unnormalized density function is
\begin{equation}\label{323}
q(\vct{\beta}, \sigma) = 
\frac{1}{\big| \frac{1}{\sigma^2} \mtx{X}\tran \mtx{X}
  + \mtx{I} \big|^{1/2}}
\exp \bigg( -\frac{\| \mtx{X} \vct{\beta} - \vct{y} \|^2 }{2\sigma^2} -  \frac{\| \vct{\beta} \| ^2}{2} \bigg),
\end{equation}
which is Gaussian in $\beta$. The expectation (and maximum) of $q$ as of function of 
$\beta$, which we denote $\bar{\beta}$ satisfies
\begin{equation}
\bar{\vct{\beta}} = \argmin_{\vct{\beta}} \| \mtx{X} \vct{\beta} - \vct{y} \|^2 + \sigma^2 \| \vct{\beta} \|^2,
\end{equation}
the ridge regression solution to the linear system~\mbox{$\mtx{X}\vct{\beta}
= \vct{y}$} with complexity parameter $\sigma^2$
\cite{hastie1}. Intuitively, for larger measurement error (larger
$\sigma^2$), the posterior mean function shrinks towards the Gaussian
process mean function, in this case $0$. The maximum
$\bar{\vct{\beta}}$ can be computed as the solution to the $m \times
m$ symmetric, positive semi-definite linear system
\begin{equation}
(\mtx{X\tran} \mtx{X} + \sigma^2 \mtx{I} ) \vct{\beta} = \mtx{X\tran} \vct{y}
\end{equation}
where the inverse of $\mtx{X}\tran\mtx{X} + \sigma^2 \mtx{I}$ can be computed using the SVD of $\mtx{X}$ computed in (\ref{437}) via the identity
\begin{equation}
\begin{split}
(\mtx{X\tran X} + \sigma^2\mtx{I})^{-1} 
& = ((\mtx{UDV\tran)\tran} \mtx{UDV\tran} + \sigma^2 \mtx{I})^{-1} \\
& = \mtx{V\tran}(\mtx{D}^2 + \sigma^2 \mtx{I})^{-1} \mtx{V}
\end{split}
\end{equation}
where $\mtx{D}$ is a diagonal $m \times m$ matrix, $\mtx{V}$ is an orthogonal $m \times m$ matrix and the columns of $\mtx{U}$, a $N \times m$ matrix, are orthonormal.
Furthermore, completing the square of $q$ in (\ref{323}), we obtain
\begin{equation}
  \begin{split}
  q(\vct{\beta}, \sigma) = 
  & \frac{1}{\big| \frac{1}{\sigma^2} \mtx{X}\tran\mtx{X} + \mtx{I} \big|^{1/2}} \\
  & \exp \bigg( -\frac{1}{2}(\vct{\beta} - \bar{\vct{\beta}})\tran
  \bigg(\frac{\mtx{X}\tran \mtx{X}}{\sigma^2} + \mtx{I} \bigg)^{-1} (\vct{\beta} -  \bar{\vct{\beta}}) \bigg)
  \end{split}
\end{equation}
which is the Gaussian 
\begin{equation}
  q(\vct{\beta}, \sigma) \sim \mathcal{N}\bigg(\bar{\vct{\beta}},
  \frac{\mtx{X}\tran \mtx{X}}{\sigma^2} + \mtx{I}\bigg).
\end{equation}
Using standard Gaussian identities, the posterior mean is given by 
\begin{equation}
\vct{\beta\tran} \vct{\phi}
\end{equation}
and the posterior variance satisfies 
\begin{equation}
\vct{\phi\tran} (\frac{1}{\sigma^2} \mtx{X\tran X} + \mtx{I})^{-1} \vct{\phi} 
= \vct{\phi\tran} \mtx{V} \frac{\sigma^2}{\mtx{D}^2 + \sigma^2} \mtx{V\tran} \vct{\phi} 
\end{equation}
where $\vct{\phi} \in \R^m$ is defined by
\begin{equation}
\vct{\phi} = [\phi_1(x) \,\, \phi_2(x) \,\, ... \,\, \phi_m(x) ]\tran.
\end{equation}
In Figure~\ref{2010} we include an illustration of the posterior
mean in weight space for a Gaussian process with 
randomly generated
data in $1$ dimension with Mat\'ern covariance kernel.

\begin{figure}[ht!]
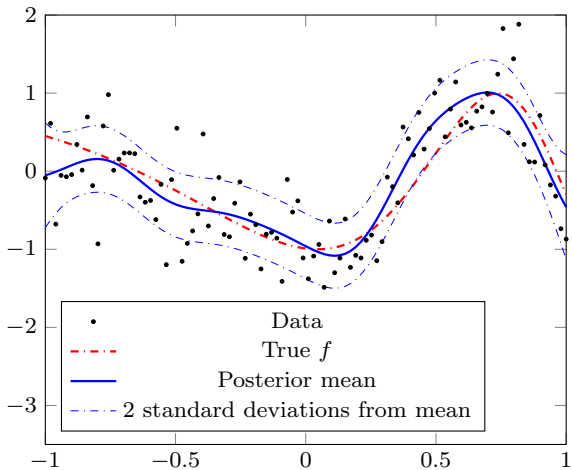

\centering

\caption{
Gaussian process regression with squared exponential kernel with $\ell = 0.2$}
\label{2010}
\end{figure}

\section{Fitting hyperparameters}\label{s550}
In certain applications, hyperparameters of the covariance function are known a priori and are chosen according to, for example, physical properties. For those problems, regression is often performed using the tools and models of the preceding sections. However in many applied environments, hyperparameters of the covariance function are learned from the data. We now describe how using KL-expansions impacts maximum likelihood and Bayesian regression models.

\subsection{Maximum likelihood}

When using KL-expansions for Gaussian processes, the maximum
likelihood approach to hyperparameter estimation involves
finding the maximum of the function
\begin{equation}\label{150}
\begin{split}
q(\vct{\beta}, \vct{\theta}, \sigma) =  
&\frac{1}{\big| \frac{1}{\sigma^2} \mtx{X_{\theta}}\tran\mtx{X_{\theta}} + \mtx{I_k} \big|^{1/2}} \\
& \exp \bigg( -\frac{\| \mtx{X_{\theta}} \vct{\beta} - \vct{y} \|^2 }{2\sigma^2} - \frac{\| \vct{\beta} \|}{2} ^2 \bigg)
\end{split}
\end{equation}
where $\vct{\theta} \in \R^d$ for some $d > 0$ are hyperparameters of
the covariance function, $\mtx{X_{\theta}}$ is the $N \times m$ matrix
(\ref{745}), and $y \in \R^N$ is the data. The entries of
$\mtx{X_{\theta}}$ will depend on the hyperparameters $\vct{\theta}$
and usually involve the recomputation of the KL-expansion via
Algorithm \ref{a1}. However, when compared to other reduced-rank
algorithms, this is not necessarily a computational bottleneck for
two main reasons. First, for problems with large amounts of data,
evaluation of KL-expansions is computationally inexpensive, $O(m^3)$
operations, compared to the cost of solving the linear system
\begin{equation}
\mtx{X_{\theta}} \vct{\beta} = \vct{y},
\end{equation}
which is $O(Nm^2)$ operations. Second, when dealing with families of
covariance kernels where evaluating KL-expansions can be costly,
eigendcompositions can be precomputed for a range of hyperparameter
values. Additionally, the determinant in $q$ can be evaluated in only
$O(m^3)$ operations using classical numerical methods~\cite{stoer}.

\subsection{Bayesian inference}

Fully Bayesian approaches to applied Gaussian process problems are
also common in practice (see, e.g., \cite{rasmus2}). For a wide range
of covariance kernels, the algorithms of this paper can substantially
reduce the oftentimes prohibitive costs of Bayesian inference. Using
Algorithm \ref{a1}, Bayesian inference is reduced to a so-called normal-normal
model with unnormalized posterior density
\begin{equation}\label{450}
\begin{split}
q(\vct{\beta}, \sigma, \vct{\theta}) 
& =  \frac{1}{\big| \frac{1}{\sigma^2} \mtx{X_{\theta}}\tran\mtx{X_{\theta}} + \mtx{I_k} \big|^{1/2}} \\
& \exp \bigg( -\frac{\| \mtx{X}_{\theta} \vct{\beta} - \vct{y} \|^2 }{2\sigma^2} - \frac{\| \vct{\beta} \| ^2}{2} \bigg) p(\vct{\theta}, \sigma)
\end{split}
\end{equation}
where $p(\vct{\theta}, \sigma)$ is some prior on hyperparameters $\vct{\theta}$
and the residual standard deviation $\sigma$. MCMC methods can be used
to sample from the posterior density $q$ via probabilistic programming tools such as 
Stan~\cite{carpenter1}. 
Additionally, since $\vct{\beta}$ in
(\ref{450}) has a Gaussian prior and Gaussian likelihood, $q$
is amenable to efficient numerical methods for inference, particularly
when the number of hyperparameters is small (see, e.g. \cite{greeng1}).

Below we provide an algorithm for the evaluation of moments of 
posteriors (\ref{450}) in which the covariance function depends on two 
parameters -- the amplitude and the timescale -- that are fit from the 
data and priors. We describe the algorithm for the model
\begin{equation}
  \begin{aligned}
y & \sim f(x) + \epsilon \\
f & \sim \mathcal{GP}(0, k(x, x')) \\
\epsilon & \sim \mathcal{N}(0, \sigma^2) 
  \end{aligned}
\end{equation}
where $k$ is a squared exponential kernel 
\begin{align}
k(x,x') = \alpha \exp \bigg( -\frac{(x-x')^2}{2\ell^2}\bigg)
\end{align}
and $\alpha$, $\sigma$, and $\ell$ are given priors
\begin{equation}
  \begin{aligned}
\alpha & \sim \cN^+(0, 3) \\
\sigma & \sim \cN^+(0, 3) \\
\ell & \sim \text{U}(0.02, 1.0).
  \end{aligned}
\end{equation}
where $\cN^+(\mu, \sigma^2)$ denotes the normal distribution $\cN(\mu,
\sigma^2)$ restricted to the non-negative reals.  We note that the
numerical efficiency of the following algorithm does not depend on the
covariance function being in the squared exponential family, nor does
it depend on the particular choices of priors. The algorithm we
describe is a generalization of Algorithm 1 in \cite{greeng1}, which
provides a numerical method for computing posterior moments of
Bayesian linear regression models. Algorithm 1 of~\cite{greeng1} performs
quadrature over a low dimensional space after analytically
marginalizing the regression coefficients.

In the following algorithm we compute posterior moments of a Bayesian
Gaussian process regression model by first discretizing the timescale
hyperparameter $\ell$ with Gaussian nodes. We then use the tools of
\cite{greeng1} to compute the posterior mean and covariance, for each 
$\ell_i$. 

\begin{algorithm0} [Reduced-rank Bayesian inference] \, 
  \label{a2}
  \begin{enumerate}
  \item Construct $n$ Gaussian nodes $\ell_i$ and weights $w_i$ on the
    interval $(0.02, 1.0)$ that will be used to discretize the kernel
    hyperparameter $\ell$.
  \item For each $\ell_i$ compute the KL-expansion corresponding to
    kernel $k$ with timescale $\ell_i$ via Algorithm~\ref{a1} and
    construct matrix $\mtx{X}_\ell$ of (\ref{745}). Note that for all
    $\alpha$, KL-expansions are identical up to a multiplicative
    constant.
 \item
 Use Algorithm 1 of \cite{greeng1} to compute moments of 
 $\vct{\beta}, \vct{\alpha}, \vct{\sigma}$ 
 with respect to density $q(\ell_i, \beta, \alpha, \sigma)$ 
 (see (\ref{450})) where $\ell_i$ is held fixed.
 \item 
 Convert conditional moments of $\vct{\beta}$ from the space of 
 coefficients in a KL-expansion to coefficients of a Legendre expansion. 
 \item
 Use conditional moments ($\ell_i$ fixed) to compute 
 posterior moments of $q$. First moments 
 of the weight space posterior are given in Legendre coefficients by
 \begin{equation}
 \bbE [\vct{c}] = \sum_{i=1}^n w_i \bbE_{\ell_i}[\vct{c}]
 \end{equation}
 where $w_i$ are Gaussian quadrature weights, $\vct{c}_i$ denotes the
 $i^{\text{th}}$ coefficient in a Legendre expansion, and
 $\bbE_{\ell_i}[\vct{c}]$ denotes the expectation of $\vct{c}$ with
 respect to density $q$ conditional on $\ell = \ell_i$.
 \end{enumerate}
\end{algorithm0}
In Section \ref{s580} we provide numerical experiments for Algorithm
\ref{a2} with the Mat\'ern covariance function in one dimension.

\section{Non-smooth covariance kernels}\label{s560}

While many commonly-used kernels are smooth (e.g. the
squared exponential, rational quadratic, and periodic kernels),
others, including Mat\'ern kernels, are not \cite{rasmus}.
The Mat\'ern kernel, $k_{\nu}$ is defined by the formula
\begin{equation}\label{785}
k_{\nu}(r) = \sigma^2 \frac{2^{1-\nu}}{\Gamma(\nu)}
\bigg(  \sqrt{2\nu}\frac{r}{\ell}\bigg)^{\nu}  K_{\nu}\bigg(\sqrt{2\nu}\frac{r}{\ell}\bigg)
\end{equation}
where $\Gamma(\nu)$ is a gamma function and $K_\nu$ is a modified
Bessel function of the second kind. For non-smooth kernels such as the
Mat\'ern, which is $\lfloor \nu \rfloor$ times differentiable at $0$,
convergence rates of the eigendecomposition using Algorithm \ref{a1}
can be slow. For such kernels, we can use another well-known numerical
scheme for computing eigendecompositions. In this scheme, we again
represent the action of the integral operator on a function as a
matrix-vector multiplication. The matrix transforms the Legendre
expansion of an inputted function to the tabulation at Gaussian nodes
of the image of that function under $\cK$. In this algorithm, we
compute elements of the matrix by using a high-order quadrature scheme
that takes advantage of the fact that the kernel is smooth away
from the origin. The eigendecomposition of that matrix is then used to
approximate the eigendecomposition of the corresponding integral
operator.

\begin{algorithm0}[Non-smooth kernels] \,
  \label{a3}
\begin{enumerate}
\item 
Construct the $n \times n$ matrix $\mtx{A}$ defined by 
\begin{equation}\label{490}
\mtx{A}_{i, j} = \sqrt{w_i} \int_{-1}^1 k(x_i, x') \, \overline{P}_{j-1}(x') \, dx' 
\end{equation}
where 
\begin{equation}
x_1,...,x_n
\end{equation}
denote the order-$n$ Legendre nodes
\begin{equation}
w_1,...,w_n
\end{equation}
the order-$n$ Gaussian weights, and $\overline{P}$ the normalized 
Legendre polynomials
(see (\ref{873})). The integral in (\ref{490}) can be
computed by, for example, representing the integral as a sum of two
integrals of smooth functions. That is,
\begin{equation}\label{510}
\begin{split}
\int_{-1}^1 k(x_i, x') \overline{P}_{j}(x') dx' 
& = \int_{-1}^{x_i} k(x_i, x') \overline{P}_{j}(x') dx' \\
& + \int_{x_i}^1 k(x_i, x') \overline{P}_{j}(x') dx'.
\end{split}
\end{equation}
We can then use Gaussian quadrature on each of the two integrals on the right hand side of (\ref{510}).

\item Compute the SVD of $\mtx{A}$. That is, find orthogonal $\mtx{U}, \mtx{V}$,
and diagonal $\mtx{D}$ such that 
\begin{equation}\label{511}
\mtx{A} = \mtx{U} \mtx{D} \mtx{V}\tran.
\end{equation}
We denote the $i$-th entry of the diagonal of $\mtx{D}$ by $\lambda_i$. 
\item Convert the columns of $\mtx{V}$ from a normalized Legendre expansion 
to an ordinary Legendre expansion via 
\begin{equation}
\hat{\mtx{V}}_{i, j} = \mtx{V}_{i, j}/ \sqrt{2 / (2(i-1) + 1)}
\end{equation}

\item Evaluate the eigenfunction approximations $v_i : [-1, 1] \to \R$ by the formula
\begin{equation}\label{342}
v_i(x) = \sum_{j=1}^n \mtx{V}_{j, i} P_{j-1}(x)
\end{equation}
for all $x \in [-1, 1]$ and $i=1,2,...,k$ where $P_j$ denotes the order-$j$ Legendre polynomial. 

\item Scale the eigenfunctions $v_i$ by the square root of the singular values. 
That is, we define $\phi_i$ by
\begin{equation}\label{352}
\phi_i(x) = \sqrt{\lambda_i} v_i(x).
\end{equation}

\item The KL-expansion of length $m \leq n$ is given by
\begin{equation}\label{723}
\hat{f}(x) = \alpha_1 \phi_1(x) + \alpha_2\phi_2(x) + ... + \alpha_m \phi_m(x)
\end{equation}
for all $x \in [-1, 1]$ where $\alpha_i \sim \cN(0, 1)$ are IID Gaussian random variables. 

\end{enumerate}
\end{algorithm0}

The convergence of this algorithm is super-algebraic for all kernels
$k = k(x, y)$ that are smooth away from $x=y$
(see~\cite{yarvin}). Specifically, for fixed $m$, the discretization
error $\alpha_n$ of (\ref{832}) decays faster than $O(1/n^{j})$ for
any $j$ (see~\cite{trefethen}).  Figure~\ref{2127all}
illustrates Algorithm \ref{a3}'s superior convergence compared to
Algorithm \ref{a1} in the approximation of eigenvalues for two
Mat\'ern kernels.  Aside from convergence rates, the description of
error control in Section \ref{s531} for KL-expansions generated using
Algorithm \ref{a1} applies in the same sense to this algorithm.

Despite Algorithm \ref{a3} possessing superior convergence 
properties than Algorithm \ref{a1}, for many practical problems there is 
little difference between the algorithms, even when the kernel is non-smooth at $0$.
Specifically, $\epsilon_n$, the $L^2$ error defined in (\ref{612}), 
has similar decay properties for non-smooth kernels when constructing
the KL-expansions via Algorithm \ref{a1} and Algorithm \ref{a3}. 
Figure~\ref{2123all} demonstrates this decay for two Mat\'ern kernels. 
The similar decay properties of these two algorithms for non-smooth kernels
is due to the fact that error $\epsilon_n$ is dominated by 
truncation error, not discretization error.


\section{Generalizations to higher dimensions}\label{s570}

Thus far we have considered only Gaussian processes over one
dimension, however in this section we focus on real-valued Gaussian 
processes over $\R^d$ for $d > 1$. For $d=2$ and $d = 3$,
applications include spatial and spatio-temporal problems~\cite{stein1, datta1}. 
Nearly all of the analytical and numerical tools described thus far for computing 
with Gaussian processes in one dimension extend naturally to higher dimensions. 
In particular, the Karhunen-Lo\`eve theorem (Theorem \ref{410}) and Algorithm 
\ref{a1} are nearly identical in $\R^d$. 

The extension of Algorithm \ref{a1} to higher dimension relies on the
discretization of functions in $\R^d$ via a tensor product of Gaussian
nodes. For the remainder of this section, we describe a numerical
algorithm for computing KL-expansions for Gaussian processes in two
dimensions.  That is, we compute eigenfunctions and eigenvalues of the
integral operator $\cK$ defined by
\begin{equation}\label{400}
\cK \mu(x) = \int_D k(x, x') \, \mu(x') \, dx'
\end{equation}
where $(x, x') \in D \times D$, and $D$ is a rectangular region in $\R^2$. 

In the two-dimensional Karhunen-Lo\`eve expansion, we represent eigenfunctions of the integral operator $\cK$ using an expansion in a tensor product of Legendre polynomials. The eigenfunction $\phi: D \to \R$ is represented as
\begin{equation}
\phi(x, y) = \sum_{i}\sum_{j} c_{ij} P_i(x)P_j(y)
\end{equation}
for all $(x, y) \in D$ where $c_{ij}$ are some real numbers. The algorithm we use for computing the eigendecomposition of integral operator $\cK$ in (\ref{400}) relies on discretizing the integral operator $\cK$ as a matrix that maps a function tabulated at two-dimensional Gaussian nodes to another function tabulated at Gaussian nodes.

\begin{algorithm0} [Eigenfunctions in two dimensions] \,
  \label{a4}
\begin{enumerate}
\item 
Construct the $n^2 \times n^2$ matrix $\mtx{A}$ in which each row and column corresponds to a point $[x_i, x_j] \in \R^2$ where $x_i$ and $x_j$ are Gaussian nodes. That is, 
\begin{equation}
\begin{split}
\mtx{A}&_{i + (n-1)j,  k + (n-1)l} = \\
& k([x_i, x_j]\tran, [x_k, x_l]\tran)\,  \sqrt{w_i w_j}\sqrt{w_k w_l}
\end{split}
\end{equation}
where 
\begin{equation}
x_1,...,x_n
\end{equation}
denote the order-$n$ Gaussian nodes
\begin{equation}
w_1,...,w_n
\end{equation}
the order-$n$ Gaussian weights.
\item Compute the diagonal form of the symmetric matrix $\mtx{A}$. That is, find the orthogonal matrix $\mtx{U}$ and the diagonal matrix $\mtx{D}$ such that
\begin{equation}\label{371}
\mtx{A} = \mtx{UDU\tran}.
\end{equation}
We denote the $i$th entry of the diagonal of $\mtx{D}$ by $\lambda_i$. 
\item Construct the $n^2 \times n^2$ matrix $\hat{\mtx{U}} = [u_i] $ defined by 
\begin{equation}
\hat{\mtx{U}}_{i + (n-1)j, k} = \mtx{U}_{i+(n-1)j, k}/\sqrt{w_i w_j}.
\end{equation}

\item Each column of $\hat{\mtx{U}}$ is a vector in $\R^{n^2}$
  denoting tabulations of an eigenfunction at the $n \times n$ tensor
  product of Gaussian nodes. We then recover the Legendre expansion in
  a tensor product of Legendre polynomials that corresponds to that
  eigenfunction. We do this by first converting the column vector
  $\hat{\mtx{U}}_i$ to an $n \times n$ matrix, $\mtx{V}_i$ and then
  evaluating the matrix $\mtx{A}_i$ of expansions coefficients defined by
\begin{equation}
\mtx{A}_i = \mtx{M} \mtx{V}_i \mtx{M}\tran
\end{equation}
where $\mtx{M}$ is the matrix of Theorem \ref{230} that maps a function tabulated
at Legendre nodes to an expansion in Legendre polynomials. 

\item $\mtx{A}_i$ is the two-dimensional eigenfunction expansion of
  the $i$-th eigenfunction of $\cK$. We evaluate the eigenfunction
  $u_i : D \to \R$ by the formula
\begin{equation}\label{341}
u_l(x, y) = \sum_{i=1}^n \sum_{j=1}^n \mtx{A}_{l_{i, j}} P_{i-1}(x)P_{j-1}(y)
\end{equation}
for all $(x, y) \in D$ and $l=1,2,...,n^2$ where $P_j$ denotes the order-$j$ Legendre polynomial. 

\item Scale the eigenfunctions $u_i$ by the square root of the eigenvalues to obtain
$\phi_i$, which we define by
\begin{equation}\label{351}
\phi_i(x, y) = \sqrt{\lambda_i} u_i(x, y)
\end{equation}
where $u_i$ is defined in (\ref{341}). 

\item The KL-expansion of length $m \leq n^2$ is given by
\begin{equation}\label{622b}
\hat{f}(x, y) = \alpha_1 \phi_1(x,y) + \alpha_2\phi_2(x,y) + ... + \alpha_m \phi_m(x,y)
\end{equation}
for all $(x, y) \in D$ where $\alpha_i \sim \cN(0, 1)$ are IID Gaussian random variables. 
\end{enumerate}

\end{algorithm0}

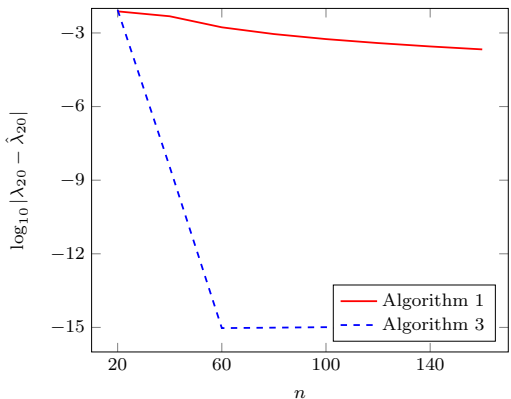
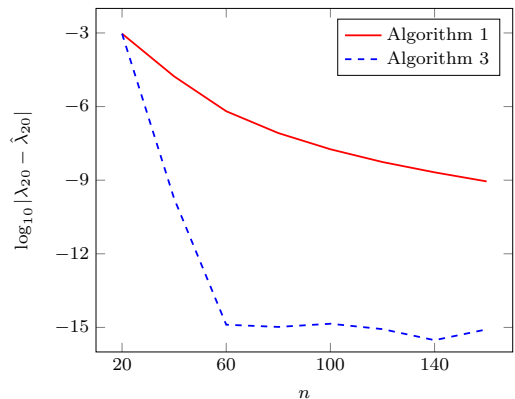
\begin{figure*}[t!]
  \centering
  \begin{subfigure}{.48\linewidth}
\begin{tikzpicture}[scale=0.8]
\centering
\begin{axis}[
    xmin=10, xmax=170,
    ymin=-16, ymax=-2.0,
    xtick={20, 60, 100, 140},
    xlabel=$n$,
    ytick={-15, -12, -9, -6, -3},
    ylabel=$\log_{10} | \lambda_{20} - \hat{\lambda}_{20} |$,
    legend pos= south east
]

\addplot[color=red,  line width=0.3mm]
    coordinates 
    {
(    20.0000000000000000,   -2.12)
(    40.0000000000000000,   -2.32)
(    60.0000000000000000,   -2.77)
(    80.0000000000000000,   -3.045)
(    100.0000000000000000,   -3.25)
(    120.0000000000000000,   -3.412)
(    140.0000000000000000,   -3.55)
(    160.0000000000000000,   -3.67)
};
    
\addplot[dashed, color=blue,  line width=0.3mm]
    coordinates 
    {
(    20.0000000000000000,   -2.06)
(    40.0000000000000000,   -8.46)
(    60.0000000000000000,   -15.028)
(    80.0000000000000000,   -15.01)
(    100.0000000000000000,   -14.99)
(    120.0000000000000000,   -15.22)
(    140.0000000000000000,   -15.43)
(    160.0000000000000000,   -15.52)
};

\legend{Algorithm \ref{a1}, Algorithm \ref{a3}}
\end{axis}
\end{tikzpicture}

\caption{Mat\'ern kernel with $\nu = 1/2, \ell = 0.2$}\label{2126}
  \end{subfigure}
  \quad
  \begin{subfigure}{.48\linewidth}
\centering
\begin{tikzpicture}[scale=0.8]
\centering
\begin{axis}[
    xmin=10, xmax=170,
    ymin=-16, ymax=-2.0,
    xtick={20, 60, 100, 140},
    xlabel=$n$,
    ytick={-15, -12, -9, -6, -3},
    ylabel=$\log_{10} | \lambda_{20} - \hat{\lambda}_{20} |$,
    legend pos=north east
]

\addplot[color=red,  line width=0.3mm]
    coordinates 
    {
(    20.0000000000000000,   -3.03)
(    40.0000000000000000,   -4.77)
(    60.0000000000000000,   -6.19)
(    80.0000000000000000,   -7.08)
(    100.0000000000000000,   -7.74)
(    120.0000000000000000,   -8.26)
(    140.0000000000000000,   -8.68)
(    160.0000000000000000,   -9.05)
};
    
\addplot[dashed, color=blue,  line width=0.3mm]
    coordinates 
    {
(    20.0000000000000000,   -3.03)
(    40.0000000000000000,   -9.75)
(    60.0000000000000000,   -14.89)
(    80.0000000000000000,   -14.98)
(    100.0000000000000000,   -14.85)
(    120.0000000000000000,   -15.07)
(    140.0000000000000000,   -15.52)
(    160.0000000000000000,   -15.08)
};
\legend{Algorithm \ref{a1}, Algorithm \ref{a3}}
\end{axis}
\end{tikzpicture}

\caption{Mat\'ern kernel with $\nu = 5/2, \ell = 0.2$}
\label{2127}
  \end{subfigure}
  \caption{Error of approximation of $\lambda_{20}$ 
when using $n$ nodes for both Algorithm \ref{a1} and Algorithm \ref{a3}.
The covariance kernel is a Mat\'ern kernel 
with $\nu = 1/2$ and $\nu = 5/2$, with $\ell = 0.2$ and the KL-expansions are of length $n$ and are defined on $[-1, 1]$}
  \label{2127all}
\end{figure*}

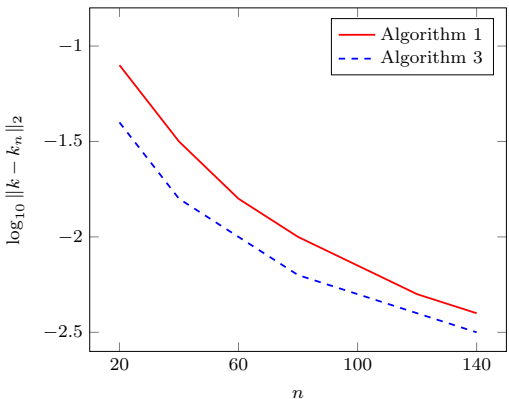
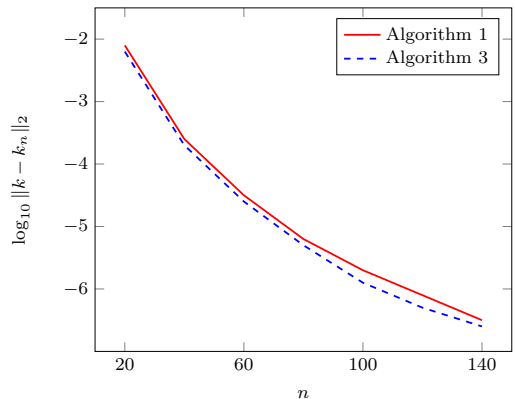
\begin{figure*}[t]
  \centering
  \begin{subfigure}{.48\linewidth}
\begin{tikzpicture}[scale=0.8]
\centering
\begin{axis}[
    xmin=10, xmax=150,
    ymin=-2.6, ymax=-0.8,
    xtick={20, 60, 100, 140},
    xlabel=$n$,
    ytick={-3,-2.5, -2, -1.5, -1, 0},
    ylabel=$\log_{10}\| k - k_n \|_2$,
    legend pos=north east
]

\addplot[color=red,  line width=0.3mm]
    coordinates 
    {
(    20.0000000000000000,   -1.1)
(    40.0000000000000000,   -1.5)
(    60.0000000000000000,   -1.8)
(    80.0000000000000000,   -2.0)
(    100.0000000000000000,   -2.15)
(    120.0000000000000000,   -2.3)
(    140.0000000000000000,   -2.4)
};
    
\addplot[dashed, color=blue,  line width=0.3mm]
    coordinates 
    {
(    20.0000000000000000,   -1.4)
(    40.0000000000000000,   -1.8)
(    60.0000000000000000,   -2.0)
(    80.0000000000000000,   -2.2)
(    100.0000000000000000,   -2.3)
(    120.0000000000000000,   -2.4)
(    140.0000000000000000,   -2.5)
};
\legend{Algorithm \ref{a1}, Algorithm \ref{a3}}
\end{axis}
\end{tikzpicture}

\caption{Mat\'ern kernel with $\nu = 1/2$, $\ell = 0.2$}
\label{2122}
  \end{subfigure}
  \quad
  \begin{subfigure}{.48\linewidth}
\centering
\begin{tikzpicture}[scale=0.8]
\centering
\begin{axis}[
    xmin=10, xmax=150,
    ymin=-7.0, ymax=-1.5,
    xtick={20, 60, 100, 140},
    xlabel=$n$,
    ytick={-6, -5, -4, -3,-2,-1,0},
    ylabel=$\log_{10} \| k - k_n \|_2$,
    legend pos=north east
]

\addplot[color=red,  line width=0.3mm]
    coordinates 
    {
(    20.0000000000000000,   -2.1)
(    40.0000000000000000,   -3.6)
(    60.0000000000000000,   -4.5)
(    80.0000000000000000,   -5.2)
(    100.0000000000000000,   -5.7)
(    120.0000000000000000,   -6.1)
(    140.0000000000000000,   -6.5)
};

\addplot[dashed, color=blue,  line width=0.3mm]
    coordinates 
    {
(    20.0000000000000000,   -2.2)
(    40.0000000000000000,   -3.7)
(    60.0000000000000000,   -4.6)
(    80.0000000000000000,   -5.3)
(    100.0000000000000000,   -5.9)
(    120.0000000000000000,   -6.3)
(    140.0000000000000000,   -6.6)
    };

\legend{Algorithm \ref{a1}, Algorithm \ref{a3}}
\end{axis}
\end{tikzpicture}

\caption{Mat\'ern kernel with $\nu = 5/2$, $\ell = 0.2$}
\label{2123}
  \end{subfigure}
  \caption{$L^2$ error (see (\ref{612})) of the effective covariance 
kernel when using $n$ nodes and KL-expansions of length $n$ for both 
Algorithm \ref{a1} and Algorithm \ref{a3}. The true covariance kernel is a Mat\'ern kernel with $\ell = 0.2$ and $\nu = 1/2, 5/2$, where expansions are defined on $[-1, 1]$.}
  \label{2123all}
\end{figure*}

The error control described in Section~\ref{s120} applies exactly to
this algorithm as well. We note that the class of algorithms described
in this paper suffers from the curse of dimensionality and the cost of
discretization of real-valued functions defined on $\R^d$ scales like
$m^d$ where $m$ is the number of discretization nodes in each
direction. Despite computational intractability in high dimensions,
eigendecompositions of operators over two and three dimensions are
still amenable to the algorithms of this paper. In the following
section we describe numerical experiments using the algorithms of this
paper for Gaussian processes over $\R$ and $\R^2$.

\section{Numerical experiments}
\label{s580}

We demonstrate the performance of the algorithms of this paper with
numerical experiments. The algorithms were implemented in Fortran and
we used the GFortran compiler on a 2.6 GHz 6-Core Intel Core i7
MacBook Pro. All examples were run in double precision arithmetic.

In this section, we focus on accuracy as measured by how well 
the true covariance kernel is approximated by the effective kernel 
implied by the KL-expansion. Under this framework, Gaussian process 
regression can be thought of as exact regression using a kernel that 
approximates to high accuracy the true kernel. 

In subsequent work, we will focus on the relationship between the 
accuracy of the effective covariance kernel and the accuracy of the approximate posterior distribution.

\subsection{Gaussian processes on the interval}
We demonstrate the performance of Algorithm \ref{a1} on randomly 
generated data on the interval \mbox{$[-1, 1] \subset \R$}. The data was 
generated according to 
\begin{equation}\label{810}
y_i = \cos(3e^{x_i}) + \epsilon_i
\end{equation}
where $x_i$ are equispaced points on $[-1, 1]$ and $\epsilon_i$ are IID 
Gaussian noise. For these experiments we used two covariance functions 
-- the squared-exponential  
\begin{equation}\label{820}
k(x,x') = \exp\bigg( -\frac{(x-x')^2}{2\ell^2}\bigg)
\end{equation}
and the Mat\'ern kernel (\ref{785}) with $\nu = 3/2$, which satisfies the identity
\begin{equation}\label{825}
k_{3/2}(r) = \bigg(1 + \frac{\sqrt{3}r}{\ell}\bigg) \exp\bigg(-\frac{\sqrt{3}r}{\ell}\bigg).
\end{equation}

In Tables \ref{3021} and \ref{3031} we demonstrate the time and accuracy
of Algorithm \ref{a1} in evaluating KL-expansions for a squared exponential
and a Mat\'ern kernel as a function of the number of discretization nodes 
used (see (\ref{790})).  The accuracy of the expansion is measured in the $L^2$ 
sense -- the columns labelled $\| k - k_n \|_2$ show the quantity
\begin{equation}\label{800}
\| k - k_n \|_2 = \Big( \int_{-1}^1 \int_{-1}^1 (k(x, y) - k_n(x, y))^2 dx dy \Big)^{1/2}
\end{equation}
where $k_n$ is the effective covariance function of the numerically 
computed order-$n$ KL-expansion using $n$ nodes and $k$ is the 
exact covariance function. Integral (\ref{800}) was computed using 
adaptive Gaussian quadrature. 

In Table \ref{3041}, we demonstrate numerical experiments of the implementation
of Algorithm \ref{a2} on the data described in (\ref{810}). Algorithm \ref{a2} 
computes posterior moments of the fully Bayesian Gaussian process model 
\begin{equation}
  \begin{aligned}
y & \sim \cN(\mtx{X}_{\ell}\vct{\beta}, \sigma^2) \\
\vct{\beta} & \sim \cN(0, \alpha \mtx{I}) \\
\alpha & \sim \cN^+(0, 3) \\
\sigma & \sim \cN^+(0, 3) \\
\ell & \sim \text{U}(0.02. 1.0)
  \end{aligned}
\end{equation}
with the $k_{3/2}$ Mat\'ern covariance function. Note that $\alpha$ 
corresponds to the magnitude of the covariance kernel, $\sigma$ the 
residual standard deviation and $\ell$ the timescale. 

In Table \ref{3041}, column $N$ corresponds to the number of data
points used, $n$ is the number of discretization points in computing
the KL-expansion (see step 2 of algorithm \ref{a2}). The column
labeled ``accuracy" denotes the maximum absolute error of posterior
expectations computed using Algorithm \ref{a2}. 
That is, ``accuracy" reports the quantity
\begin{equation}
\max \{\| \vct{\beta} - \vct{\hat{\beta}} \|_{\infty}, | \alpha - \hat{\alpha} |, | \ell - \hat{\ell} |, | \sigma - \hat{\sigma} |
\}
\end{equation}
where $\vct{\beta}$ is the true posterior mean and $\vct{\hat{\beta}}$ denotes 
the approximation using Algorithm \ref{a2}. Similarly, $\hat{\alpha}, \hat{\ell}, 
\hat{\sigma}$ denote the approximations to the exact parameter values 
$\alpha, \ell, \sigma$.
The accuracy reported depends on the number of nodes used in the 
quadrature and the smoothness of 
the posterior densities being integrated.

In Figure \ref{78abc} we report the accuracy of the posterior mean and standard
deviation for Gaussian process regression using the data-generating process
of (\ref{810}) with $N=100$ data points. We compute the ground truth using 
a dense $O(N^3)$ algorithm. $L^2$ errors were computed by tabulating posterior
means and standard deviations at $200$ equispaced nodes on $[-1, 1]$.

\begin{table*}
\parbox{.45\linewidth}{
\centering
\begin{tabular}{ccc}
   $n$ & time (ms)  & $ \| k - k_n \|_2$ \\
  \hline
  $5$   & $0.01$ & $0.40 \times 10^{0}$ \\
  $10$ & $0.02$ & $0.66 \times 10^{-1}$ \\
  $15$ & $0.04$ & $0.56 \times 10^{-2}$ \\
  $20$ & $0.08$ & $0.25 \times 10^{-3}$ \\
  $25$ & $0.13$ & $0.71 \times 10^{-5}$ \\
  $30$ & $0.20$ & $0.13 \times 10^{-6}$ \\
  $35$ & $0.28$ & $0.17 \times 10^{-8}$ \\
  $40$ & $0.36$ & $0.17 \times 10^{-10}$ \\
  $45$ & $0.44$ & $0.12 \times 10^{-12}$ \\
  $50$ & $0.54$ & $0.11 \times 10^{-13}$ 
\end{tabular}
\caption{\em KL-expansion accuracy and computation times for a one-dimensional Gaussian process with squared exponential kernel, $\ell = 0.2$, using Algorithm \ref{a1}
with $n$ nodes and an order-$n$ KL-expansion.}
\label{3021}
}
\hfill
\parbox{.45\linewidth}{
\centering
\begin{tabular}{ccc}
   $n$ & time (ms)  & $ \| k - k_n \|_2$ \\
  \hline
  $10$ & $0.03$ & $0.12 \times 10^{0}$ \\
  $15$ & $0.06$ & $0.43 \times 10^{-1}$ \\
  $20$ & $0.10$ & $0.18 \times 10^{-1}$ \\
  $25$ & $0.14$ & $0.89 \times 10^{-2}$ \\
  $30$ & $0.21$ & $0.49 \times 10^{-2}$ \\
  $35$ & $0.31$ & $0.29 \times 10^{-2}$ \\
  $40$ & $0.42$ & $0.18 \times 10^{-2}$ \\
  $45$ & $0.63$ & $0.12 \times 10^{-2}$ \\
  $50$ & $0.68$ & $0.86 \times 10^{-3}$ \\
  $55$ & $0.85$ & $0.62 \times 10^{-3}$ 
 \end{tabular}
\caption{\em KL-expansion accuracy and computation times for a one-dimensional Gaussian process with Mat\'ern kernel,  $\ell = 0.2$, using Algorithm \ref{a1}
with $n$ nodes and an order-$n$ KL-expansion.}
\label{3031}
}

\vspace{0.5cm}

\centering
\begin{tabular}{cccc}
   $N$ & $n$ & accuracy  & total time (s) \\
  \hline
  $10$ & $140$ & $0.11 \times 10^{-3}$ & 0.21 \\
  $100$ & $140$ & $0.52 \times 10^{-3}$ & 0.25 \\
  $1,000$ & $140$ & $0.10 \times 10^{-2}$ & 0.63 \\
  $10,000$ & $140$ & $0.15 \times 10^{-2}$ & 2.03 \\
  $100,000$ & $100$ & $0.33 \times 10^{-2}$ & 37.2 
\end{tabular}
\caption{\em Accuracy and timings of Bayesian inference using 
Algorithm \ref{a2} with $n$ nodes and KL-expansions of order $n$
with fitting of residual variance, timescale, and magnitude with Mat\'ern 3/2 kernel.}
\label{3041}
\end{table*}

\begin{table*}
\centering
\begin{tabular}{ccc}
   $n$ & total time (s)  & $ \| k - k_n \|_2$ \\
  \hline
  $10^2$ &  $0.004$ & $0.033 \times 10^{0}$ \\
  $12^2$ &  $0.008$ & $0.93 \times 10^{-2}$ \\
  $15^2$ &  $0.02$ & $0.11 \times 10^{-2}$ \\
  $17^2$ &  $0.04$ & $0.2 \times 10^{-3}$ \\
  $20^2$ &  $0.21$ & $0.49 \times 10^{-4}$ 
\end{tabular}

  \caption{\em KL-expansion accuracy and computation times for a 
  two-dimensional Gaussian process with squared exponential kernel 
  and $\ell = 0.25$ using Algorithm \ref{a4} with $n$ nodes
  and order-$n$ KL-expansions.}
\label{3001}
\vspace{0.5cm}

\centering
\begin{tabular}{ccccccc}
   $N$ & $n$ & $ \| k - k_n \|_2$ & KL time (s)  & regression time (s) & total time (s)\\
  \hline
  $2,500$ & $20^2$ &  $0.5 \times 10^{-4}$ & 0.05 & 0.15 & 0.20 \\
  $6,400$ & $20^2$ &  $0.5 \times 10^{-4}$ & 0.05 & 0.30 & 0.35 \\
  $10,000$ & $20^2$ &  $0.5 \times 10^{-4}$ & 0.05 & 0.51 & 0.56 \\
  $90,000$ & $20^2$ & $0.5 \times 10^{-4}$ & 0.05 & 2.38 & 2.43 \\
  $160,000$ & $20^2$ & $0.5 \times 10^{-4}$ & 0.05 & 7.07 & 7.12 
\end{tabular}

 \caption{\em KL-expansion accuracy and regression compute times 
 for two-dimensional Gaussian process regression with $N$ data points on the unit 
 square with squared exponential kernel (\ref{610}) and $\ell = 0.25$. We 
 used Algorithm \ref{a2} with an order-$n$ KL-expansions
 and accuracy is measured as $\|k - k_n\|_2$ where $k_{n}$ is the effective 
 covariance kernel.}
\label{3011}
\end{table*}

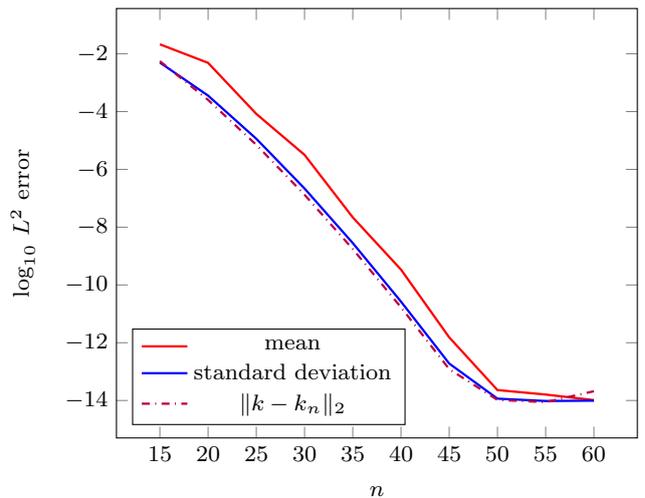
\begin{figure}[t!]
\begin{tikzpicture}
  \centering
\begin{axis}[
    xtick={15, 20, 25, 30, 35, 40, 45, 50, 55, 60},
    xlabel=$n$,
    ytick={-14, -12, -10, -8, -6, -4, -2, 0},
    ylabel=$\log_{10}$ $L^2$ error,
    legend pos=south west
]

\addplot[color=red,  line width=0.3mm]
    coordinates 
    {
(   15.0000000000000000,   -1.6707521468830315)
(   20.0000000000000000,   -2.3130802215376129)
(   25.0000000000000000,   -4.0793685436063809)
(   30.0000000000000000,   -5.4978359139345292)
(   35.0000000000000000,   -7.6625110487254338)
(   40.0000000000000000,   -9.4691067314842421)
(   45.0000000000000000,  -11.8150697125802608)
(   50.0000000000000000,  -13.6391234430281099)
(   55.0000000000000000,  -13.7944467474353374)
(   60.0000000000000000,  -13.9869584825776041)

};
    
\addplot[color=blue,  line width=0.3mm]
    coordinates 
    {
(   15.0000000000000000,   -2.3033265245125656)
(   20.0000000000000000,   -3.4466989218576360)
(   25.0000000000000000,   -4.9408777177235956)
(   30.0000000000000000,   -6.6623070788755898)
(   35.0000000000000000,   -8.5514206978709186)
(   40.0000000000000000,  -10.5777897764222093)
(   45.0000000000000000,  -12.7206571892954923)
(   50.0000000000000000,  -13.9365896824701352)
(   55.0000000000000000,  -14.0171512918112544)
(   60.0000000000000000,  -14.0064607406175430)
};

\addplot[color=purple,  dash dot, line width=0.3mm]
    coordinates 
    {
(   15.0000000000000000,   -2.2532726450051515)
(   20.0000000000000000,   -3.5959109096746613)
(   25.0000000000000000,   -5.1511217498738047)
(   30.0000000000000000,   -6.8785947381594879)
(   35.0000000000000000,   -8.7582542930965932)
(   40.0000000000000000,  -10.7752579210933206)
(   45.0000000000000000,  -12.9139915224545856)
(   50.0000000000000000,  -13.9801704427381477)
(   55.0000000000000000,  -14.0575794673647962)
(   60.0000000000000000,  -13.6792095404350142)
};

\legend{mean, standard deviation, $\|k - k_n\|_2$}
\end{axis}
\end{tikzpicture}
\caption{$L^2$ (RMS) error of Gaussian process posterior mean and standard deviation for various numbers of basis functions with squared exponential kernel ($\ell = 0.2$). Algorithm \ref{a1} was used to construct the $n$ basis functions with $n$ discretization nodes. Data was generated via $y_i = \cos(3e^{x_i}) + \epsilon_i$ for $N = 100$ uniformly distributed $x_i$ on $[-1, 1]$ and $\epsilon_i$ iid Gaussian noise. $\| k - k_n \|_2$ denotes the $L^2$ accuracy of the effective kernel (see (\ref{800})). }
\label{78abc}
  \end{figure}

\begin{figure*}[t!]
  \centering
  \begin{subfigure}{.3\linewidth}
\begin{tikzpicture}[scale=0.6]
\centering
\begin{axis}[
    xmin=8, xmax=32,
    ymin=-14, ymax=0.0,
    xtick={5, 10, 15, 20, 25, 30, 35},
    xlabel=$m$,
    ytick={-14, -12, -10, -8, -6, -4, -2, 0},
    ylabel=$\log_{10} \| \vct{\bar{y}}_{\text{true}} - \vct{\bar{y}}_{\text{approx}} \|_{\infty}$,
    legend pos=south west
]

\addplot[color=red,  line width=0.3mm]
    coordinates 
    {
(    10.0,   -1.7887)
(    15.0,   -1.71)
(    20.0,   -1.71)
(    25.0,   -1.71)
(    30.0,   -1.71)
};
    
\addplot[color=blue,  line width=0.3mm]
    coordinates 
    {
(    10.0,    -1.9505)
(    15.0,   -4.7987)
(    20.0,   -7.6986)
(    25.0,   -11.7084)
(    30.0,  -13.5346)
};
\legend{Hilbert-GP, Algorithm \ref{a1}}
\end{axis}
\end{tikzpicture}

\caption{$\ell = 0.25$}
\label{78a}
  \end{subfigure}
  \quad
  \begin{subfigure}{.3\linewidth}
\begin{tikzpicture}[scale=0.6]
\centering
\begin{axis}[
    xmin=8, xmax=42,
    ymin=-14, ymax=0.0,
    xtick={5, 10, 15, 20, 25, 30, 35, 40},
    xlabel=$m$,
    ytick={-14, -12, -10, -8, -6, -4, -2, 0},
    legend pos=south west
]
    
\addplot[color=red,  line width=0.3mm]
    coordinates 
    {
(    10.0,     -1.6355)
(    15.0,   -2.1874)
(    20.0,    -2.3841)
(    25.0,    -2.3841)
(    30.0,    -2.3841)
(    35.0,    -2.3841)
(    40.0,    -2.3841)
};
\addplot[color=blue,  line width=0.3mm]
    coordinates 
    {
(    10.0,   -1.3930)
(    15.0,   -3.0205)
(    20.0,   -5.4771)
(    25.0,  -8.5986)
(    30.0,   -12.2042)
(    35.0,   -13.1793)
(    40.0,   -13.0614)
};
\legend{Hilbert-GP, Algorithm \ref{a1}}
\end{axis}
\end{tikzpicture}

\caption{$\ell = 0.2$}
\label{78b}
  \end{subfigure}  
  \begin{subfigure}{.3\linewidth}
\centering
\begin{tikzpicture}[scale=0.6]
\centering
\begin{axis}[
    xmin=8, xmax=72,
    ymin=-14, ymax=0.0,
    xtick={10, 20, 30, 40, 50, 60, 70},
    xlabel=$m$,
    ytick={-14, -12, -10, -8, -6, -4, -2, 0},
    legend pos=south west
]

\addplot[color=red,  line width=0.3mm]
    coordinates 
    {
(    10.0,    -0.7930)
(    15.0,   -0.7994)
(    20.0,   -1.0465)
(    25.0,   -1.9570)
(    30.0,   -2.8126)
(    35.0,    -3.8018)
(    40.0,   -4.8915)
(    45.0,   -6.3667)
(    50.0,    -7.5892)
(    55.0,   -8.7219)
(    60.0,   -8.7481)
(    65.0,   -8.7479)
(    70.0,   -8.7479 )
};
   
\addplot[color=blue,  line width=0.3mm]
    coordinates 
    {
(    10.0,    -0.4412)
(    15.0,   -0.6904)
(    20.0,   -1.4800)
(    25.0,   -2.7413)
(    30.0,   -3.2353)
(    35.0,   -5.9360)
(    40.0,   -6.7380)
(    45.0,   -8.4460)
(    50.0,  -10.0034)
(    55.0,  -11.8965)
(    60.0,   -12.5584)
(    65.0,  -12.8307)
(    70.0,  -12.6745)
};
\legend{Hilbert-GP, Algorithm \ref{a1}}
\end{axis}
\end{tikzpicture}

\caption{$\ell = 0.1$}
\label{78c}
  \end{subfigure}
\caption{The $\log_{10}$ $L^{\infty}$-norm of the difference between the exact posterior mean 
and two basis function approximations --  Algorithm \ref{a1} and the approach of \cite{solin1} (Hilbert-GP). Gaussian process regression was performed on $[-1, 1]$ with 
squared exponential kernel with $\ell = 0.25, 0.2, 0.1$. Data was randomly generated 
according to $y_i = \sin(2x_i) + \epsilon_i$ for $i = 1,...,100$ where $x_i$ were 
generated uniformly at random on $[-1, 1]$.}
\label{78all}
\end{figure*}
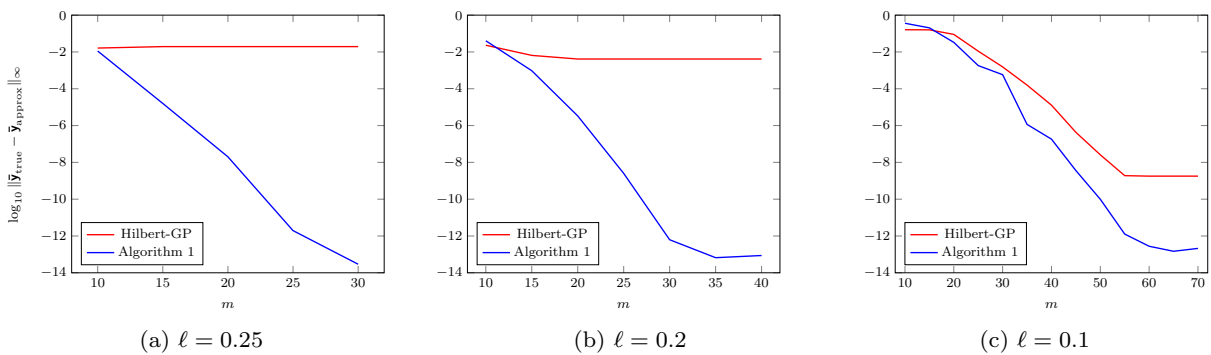

In Figure \ref{78all}, we illustrate the accuracy of our method and the 
method of \cite{solin1} for various numbers of basis functions. 
We perform Gaussian process regression 
on data simulated according to 
\begin{equation}
y_i = \sin(2x_i) + \epsilon_i
\end{equation}
for $i = 1,...,100$ where $x_i$ were 
generated uniformly at random on $[-1, 1]$ and $\epsilon_i$ were 
generated iid according to $\epsilon_i \sim \mathcal{N}(0, 1)$. 
We used a squared-exponential kernel with several different timescales.
We computed accuracy of each method by comparing to
the same calculation using a straightforward, 
exact $O(N^3)$ algorithm. We use the code 
of~\cite{solin_github} as the implementation of 
the method of \cite{solin1}.

The approach of \cite{solin1} has the 
desirable feature that the basis functions they use
are virtually
free to compute. However, the approximations used
to construct their basis function expansions can 
result in loss of accuracy, especially for kernels
without small timescale. 
Using the KL-expansion approach of this paper, we 
achieve high accuracy in the approximation of
basis functions by using high-order quadrature.
These methods do require an extra computational task -- the evaluation of an
eigendecomposition. However, for commonly-used 
kernels in $1$ and $2$ dimensions, evaluation of 
KL-expansions is negligible compared to 
subsequent regression tasks for problems 
with moderate amounts of data
(see Tables \ref{3021}, \ref{3031}, \ref{3001}).

The accuracy of \cite{solin1} and 
Algorithm \ref{a1} for various numbers of basis 
functions is illustrated in Figure \ref{78all} 
for several timescales.

\subsection{Gaussian processes in two dimensions}
For constructing KL-expansions in two dimensions,
 we implemented  Algorithm \ref{a4} and 
tested timing and accuracy on randomly generated 
data on the unit square in $\R^2$. We used the squared 
exponential covariance kernel 
\begin{equation}\label{610}
k(x,x') = \exp\bigg( -\frac{\|x-x'\|^2}{2\ell^2}\bigg)
\end{equation}
with $\ell = 0.25$. 
The data was defined on a square grid at the points 
$\{(x_{1,i}, x_{2, j})\}$ on the square $[-1, 1]^2$ where $x_{1, 1},...,x_{1,N}$ 
and $x_{2, 1},...,x_{2, N}$ are equispaced points on the interval $[-1, 1]$. 
The dependent variable $y_{i,j}$ was randomly generated according to 
\begin{equation}
y_{i} = -x_{2, i} + \sin(6x_{1, i}) + \epsilon_i
\end{equation}
where
\begin{equation}
\epsilon \sim \cN(0, \mtx{I}).
\end{equation}
In Figures \ref{705}, \ref{710}, and \ref{715}, we provide plots of the ground truth, 
the observed values, and the recovered posterior mean estimate. 

In Table \ref{3001} we demonstrate the performance of Algorithm
\ref{a4} as a function of the total number of nodes $n$. The column
labeled $\| k - k_n \|_2 $ measures the accuracy of the order-$n$ KL
expansion evaluated using $n$-nodes in the following sense
\begin{equation}\label{795}
\| k - k_n \|_2 = \Big( \int_D \int_D (k(x, y) - k_n(x, y))^2 dx \, dy \Big)^{1/2}
\end{equation}
where $D = [-1, 1]^2$ and $k_n$ is the effective covariance function of the numerically computed order-$n$ KL-expansion. Integral (\ref{795}) was computed using Gaussian quadrature. 

In Table \ref{3011} we provide timings and accuracy for computing
Gaussian process posterior mean estimates where all hyperparameters
are fixed. The column denoted $N$ is the number of data points and $n$
represents the number of nodes used for computing KL-expansions in
Algorithm \ref{a4}. ``KL time (s)" shows the total amount of time used
to compute KL-expansions and "regression time (s)" denotes the total time for
computing posterior mean and covariance estimates after computing
KL-expansions. This time includes constructing matrix $\mtx{X}$
of~\eqref{745} and computing the ridge regression.

In addition to the Fortran implementations that we used for the
numerical results of this section, we implemented Algorithm \ref{a1}
in Python and have made the code publicly available at
\begin{quote}
  \mbox{\url{https://github.com/pgree/kl_exps}}.
\end{quote}
The purpose of the
Python code is to provide a user-friendly implementation of Algorithm
\ref{a1} in a commonly-used language that can serve as a template for
general Gaussian process regression tasks.

\begin{figure*}[t]
  \centering
  \begin{subfigure}{.3\textwidth}
    \centering
    \includegraphics[width=.95\textwidth]{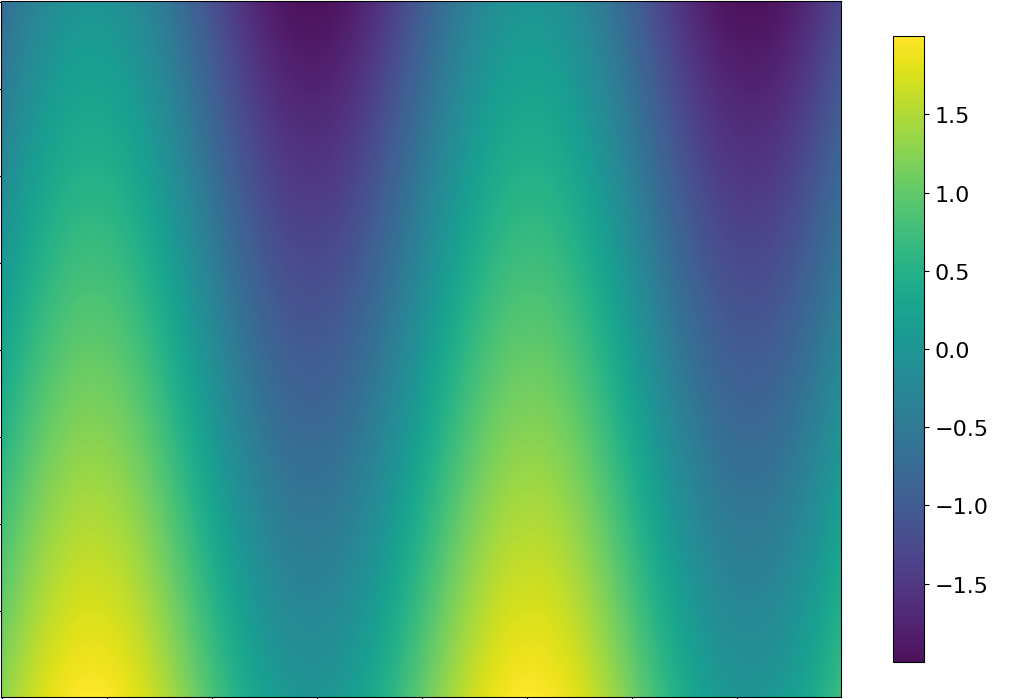}
    \caption{Ground truth}
    \label{705}
  \end{subfigure}
  \quad
  \begin{subfigure}{.3\textwidth}
    \centering
    \includegraphics[width=.95\textwidth]{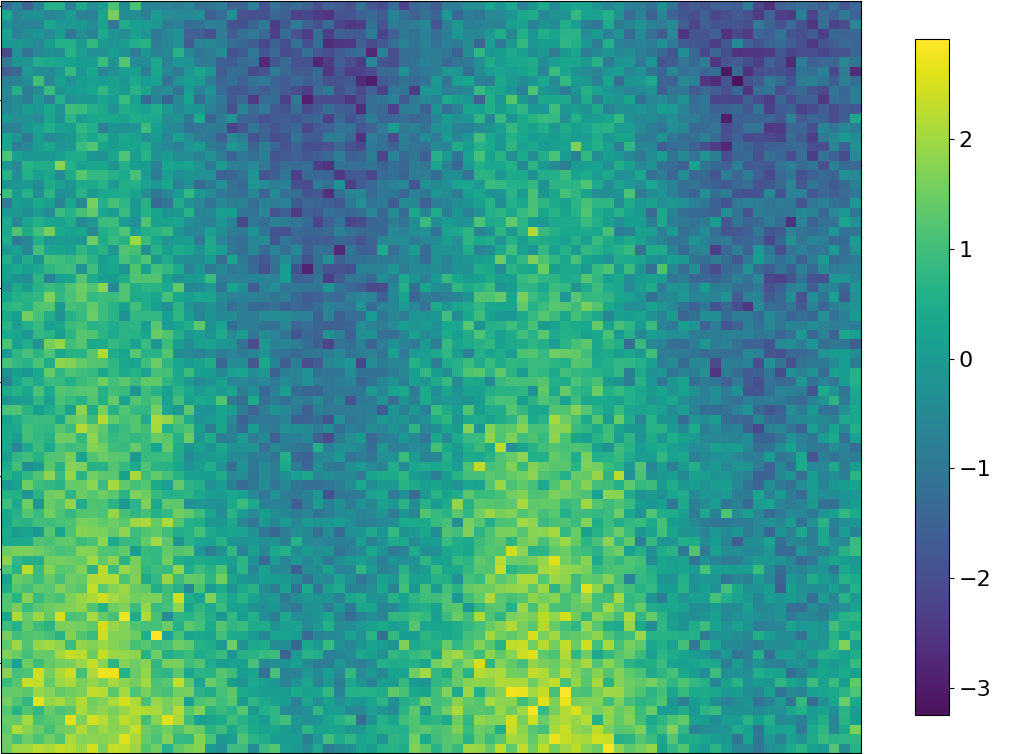}
    \caption{Data}
    \label{710}
  \end{subfigure}
  \quad
  \begin{subfigure}{.3\textwidth}
    \centering
    \includegraphics[width=.95\textwidth]{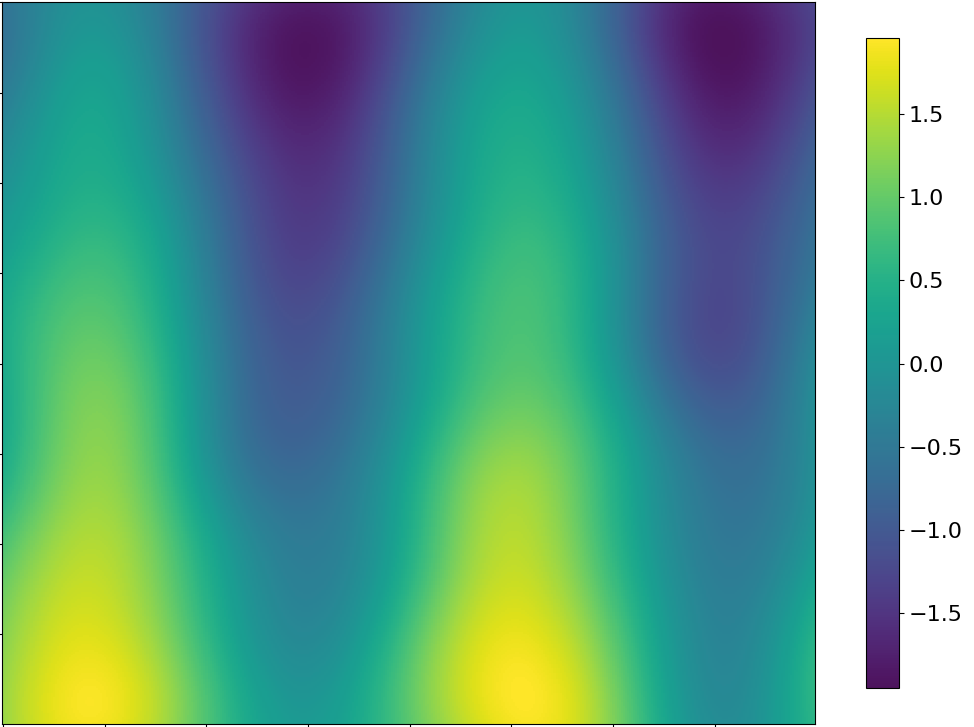}
    \caption{Posterior mean}
    \label{715}
  \end{subfigure}
  \caption{A demonstration of the performance of Algorithm~\ref{a4} on
    a Gaussian process over two dimensions with randomly generated
    data. The data is defined on a square grid on $[-1, 1]^2$ and
    $y_i = -x_{2, i} + \sin(6x_{1, i}) + \epsilon_i$ where
    $\epsilon_i$ is IID Gaussian noise.}
\end{figure*}

\section{Conclusions}
\label{sec:conclusions}
In this paper we introduce a class of numerical methods for converting
a Gaussian process \mbox{$f \sim \mathcal{GP}(0, k)$} on a
user-defined rectangular region of $\R^d$ into a KL-expansion of the
form
\begin{equation}\label{632}
\hat{f}(x) = \alpha_1f_1(x) + ... + \alpha_m f_m(x)
\end{equation}
where $\alpha \sim \cN(0, 1)$ are IID and $f_i$ are fixed basis functions
computed once during precomputation. The KL-expansion 
has several qualities that make it attractive for computationally 
demanding Gaussian process problems. 

\begin{enumerate}
\item The KL-expansion is optimal in the $L^2$ sense. Specifically,
for any order-$n$ basis function representation of a Gaussian process, 
the KL-expansion has an effective covariance kernel that best approximates
the true kernel in the $L^2$ sense. This allows for highly accurate and
compressed representations of Gaussian processes. For example, a Gaussian
process on the interval $[-1, 1]$ with squared exponential kernel 
\begin{equation}
k(x,x') = \exp\bigg( -\frac{\|x-x'\|^2}{2\ell^2}\bigg)
\end{equation}
with $\ell = 0.1$, can be approximated to an accuracy of better than 
$10^{-3}$ with an expansion of $25$ basis functions.

\item KL-expansions can be computed directly and efficiently using 
well-known high-order algorithms for discretizing integral operators. 
For smooth kernels, convergence of these algorithms is super-algebraic. 
For kernels with $j$ continuous derivatives, convergence is no worse
than $O(n^{-j-1})$ where $n$ is the number of discretization nodes. 

\item 
Efficient statistical inference can be facilitated with KL-expansions. 
When viewed as a weight-space problem, the canonical Gaussian 
process regression  
is converted to a ridge regression in the space of expansion
coefficients, where the number of coefficients is often significantly
smaller than the number of data points.  We also introduce an algorithm
for rapidly evaluating posterior moments of Bayesian models. 

\end{enumerate}

The methods of this paper will likely generalize naturally to some
families of non-Gaussian
stochastic processes, such as stable distributions. A stable distribution is
one where a linear combination of independent copies of the
distribution follows the same distribution as the original, up to
scale and location parameters~\cite{nolan1}.  For example, suppose
that we replace $\alpha_1,...,\alpha_m$ in (\ref{632}) with
uncorrelated stable distributions such that
\begin{equation}
\bbE[\alpha_i] = 0 \qquad \text{and} \qquad \bbE[\alpha_i \alpha_j] = \delta_{ij}.
\end{equation}
Then the KL-expansion
\begin{equation}
f(x) = \sum_{i=1}^m \alpha_i f_i(x)
\end{equation}
is in the same family of distributions as the $\alpha_i$ and satisfies
\begin{equation}
\bbE[f(x)] = 0 \qquad \text{and} \qquad \bbE[f(x)f(y)] \approx k(x, y).
\end{equation}
As a result, nearly all the numerical and analytical results of this paper
generalize naturally to stable processes. Analytic and numerical
investigations on this line of work are currently underway.

There are two main failure modes to the schemes of this paper. First,
covariance kernels that are less smooth (i.e. they have more slowly
decaying power spectra) require more terms in a KL-expansion for a
given level of accuracy. Second, the methods of this paper suffer from
the usual curse of dimensionality. For a given kernel of a Gaussian
process over $\R^d$, the number of terms needed in a KL-expansion for
a given level of accuracy grows like $m^d$ where $m$ is the number of
terms required in one dimension. 
Due to these drawbacks, the numerical
methods we describe are most useful in one and two dimensions, or in
three-dimensional problems with smooth kernels.

\section{Acknowledgements}
The authors are grateful to Paul Beckman, Dan Foreman-Mackey, Jeremy
Hoskins, Manas Rachh, and Vladimir Rokhlin for helpful discussions.
The first author is supported by Alfred P. Sloan Foundation. The second 
author is supported in part by
the Office of Naval Research under award
numbers~\#N00014-21-1-2383 and the Simons Foundation/SFARI (560651, AB).

\bibliographystyle{apalike}
\bibliography{refs}

\appendix
\section{Legendre polynomials}
\label{sec:legendre}

We now provide a brief overview of Legendre polynomials
and Gaussian quadrature~\cite{abram}. For a more in-depth analysis of these
tools and their role in (numerical) approximation
theory see, for example,~\cite{trefethen}.

In accordance with standard practice, we denote by~$P_i: [-1, 1] \to \R$ the Legendre
polynomial of degree~$i$ defined by the three-term recursion
\begin{equation}
P_{i+1}(x) = \frac{2i + 1}{i + 1} \, x \, P_{i}(x) - \frac{i}{i + 1} P_{i-1}(x)
\end{equation}
with initial conditions
\begin{equation}
P_0(x) = 1 \qquad \text{and} \qquad P_1(x) = x.
\end{equation}
Legendre polynomials are orthogonal on $[-1, 1]$ and satisfy
\[ 
\int_{-1}^{1} P_i(x) \, P_j(x) \, dx = 
\begin{cases} 
      0 & i \neq j, \\
      \frac{2}{2i+1} & i = j.
   \end{cases}
\]
We denote the $L^2$ normalized Legendre polynomials, $\overline{P}_i$, which
are defined by
\begin{equation}\label{873}
\overline{P}_i(x) = \sqrt{\frac{2i + 1}{2}} P_i(x).
\end{equation}
For each~$n$,  the Legendre polynomial~$P_n$ has~$n$ distinct roots which we
denote in what follows  by~$x_1,...,x_n$. Furthermore, for all $n$, there exist $n$
positive real numbers~$w_1,...,w_n$ such that for any polynomial~$p$ of
degree $ \leq 2n - 1$,
\begin{equation}
\int_{-1}^1 p(x) \, dx = \sum_{i=1}^n w_i \, p(x_i).
\end{equation}
The roots~$x_1,\ldots,x_n$ are usually referred to as order-$n$ Gaussian
nodes and~$w_1,...,w_n$ the associated Gaussian quadrature weights.
Classical Gaussian quadratures such as this are associated with many
families of orthogonal polynomials: Chebyshev, Hermite, Laguerre,
etc. The quadratures we mention above, associated with Legendre
polynomials, provide a high-order method for discretizing
(i.e. interpolating) and integrating square-integrable functions on a
finite interval. Legendre polynomials are the natural orthogonal
polynomial basis for
square-integrable functions on the interval~$[-1,1]$, and the
associated interpolation and quadrature formulae provide nearly
optimal approximation tools for these functions, even if they are not,
in fact, polynomials.

The following well-known lemma regarding interpolation using Legendre
polynomials will be used in the numerical schemes discussed in this
paper. A proof can be found in~\cite{stoer}, for example.
\begin{theorem}\label{230}
  Let $x_1,...,x_n$ be the order-$n$ Gaussian nodes and $w_1,...,w_n$
  the associated order-$n$ Gaussian weights. Then there exists an $n
  \times n$
  matrix~$\mtx{M}$ that maps a function tabulated at these Gaussian
  nodes to the corresponding Legendre expansion, i.e. the
  interpolating polynomial expressed in terms of Legendre
  polynomials. That is to say, defining~$\vct{f}$ by
\begin{equation}
  \vct{f} = \lp f(x_1) \cdots f(x_n) \rp\tran,
\end{equation}
the vector
\begin{equation}
\vct{\alpha} = \mtx{M}\vct{f}
\end{equation}
are the coefficients of the order-$n$ Legendre expansion~$p$ such that
\begin{equation}
  \begin{aligned}
    p(x_j) &= \sum_{i=1}^n \alpha_i  \, P_{i-1}(x_j)\\
    &= f(x_j),
  \end{aligned}
\end{equation}
where $\alpha_i$ denotes the $i$th entry of the vector $\vct{\alpha}$.
\end{theorem}
From a computational standpoint, algorithms for efficient evaluation of Legendre 
polynomials and Gaussian nodes and weights are  
available in standard software packages (e.g. \cite{driscoll2014}). 
Furthermore, the entries of the matrix $\mtx{M}$ can be computed 
directly via $M_{i, j} = w_j P_{i-1}(x_j)$.

\end{document}